\documentclass{lmcs} %


\titlecomment{%
  This paper is an extended version of the paper with the same title which appeared in CONCUR'25~\cite{PaulPS0TY25}.
}

\usepackage{hyperref}
\usepackage{comment}

\usepackage[utf8]{inputenc}
\usepackage{tikz}
\usetikzlibrary{automata, positioning, arrows,shapes.geometric, calc}
\usepackage{nicefrac}
\usepackage[justification=centering]{caption}
\usepackage{subcaption}
\usetikzlibrary{arrows.meta}
\usepackage{amsmath}
\usepackage{amsthm}
\usepackage{amssymb}
\usepackage{diagbox}

\usepackage[capitalise]{cleveref}
\crefname{thm}{Theorem}{Theorems}   %
\Crefname{thm}{Theorem}{Theorems}  %
\crefname{lem}{Lemma}{Lemmas}   %
\Crefname{lem}{Lemma}{Lemmas}  %
\crefname{defi}{Definition}{Definitions}   %
\Crefname{defi}{Definition}{Definitions}  %
\crefname{cor}{Corollary}{Corollaries}   %
\Crefname{cor}{Corollary}{Corollaries}  %
\crefname{exa}{Example}{Examples}   %
\Crefname{exa}{Example}{Examples}  %
\crefname{obs}{Observation}{Observations}   %
\Crefname{obs}{Observation}{Observations}  %

\crefname{enumi}{}{}
\Crefname{enumi}{}{}

\usepackage{tabularx}

\usepackage{xspace}
\usepackage{comment}

\usepackage{complexity}
\newclass{\PTIME}{PTIME}

\usepackage{makecell}

\usepackage{xifthen}

\usepackage{multirow}
 \usepackage[export]{adjustbox}

\newcommand{\alphabet}{\Sigma}
\newcommand{\trans}{\delta}
\newcommand{\Trans}{\Delta}

\usepackage{etoolbox}
\AtEndEnvironment{example}{\lipicsEnd}

\newcommand{\emptyword}{\epsilon}

\newcommand{\infwords}{\alphabet^\omega}

\newcommand{\prio}{\rho}

\newcommand{\Ll}{\mathcal{L}}
\newcommand{\Lang}{\mathcal{L}}
\newcommand{\machlang}[1]{\Lang({#1})}

\newcommand{\buchi}{B\"uchi~}

\newcommand{\xra}[1]{\xrightarrow{#1}}

\newcommand{\lmb}{\lambda}

\newcommand{\incl}{\subseteq}

\newcommand{\nats}{\mathbb{N}\xspace}
\newcommand{\ints}{\mathbb{Z}\xspace}
\newcommand{\positives}{\mathbb{Z}_+\xspace}
\newcommand{\rats}{\mathbb{Q}\xspace}
\newcommand{\reals}{\mathbb{R}\xspace}

\newcommand{\positiverats}{\mathbb{Q}_+\xspace}

\newcommand{\nfa}{NFA\xspace}
\newcommand{\dfa}{DFA\xspace}

\newcommand{\pfa}{PFA\xspace}

\newcommand{\nba}{NBA\xspace}

\newcommand{\mach}[1]{\mathcal{#1}\xspace}
\renewcommand{\A}{\mach{A}}
\newcommand{\npa}{NPA\xspace}

\newcommand{\fnfa}{FNFA\xspace}
\newcommand{\unfa}{UFA\xspace}

\newcommand{\acc}[1]{\operatorname{ACC}(#1)\xspace}
\newcommand{\accsub}[2]{\operatorname{ACC}_{#1}(#2)\xspace}

\newcommand{\finamb}{finitely-ambiguous\xspace}
\newcommand{\kamb}[1]{$#1$-ambiguous\xspace}
\newcommand{\unamb}{unambiguous\xspace}
\newcommand{\mr}{positively resolvable\xspace}
\newcommand{\lmr}[1]{%
\ifthenelse{\isempty{#1}}%
{$\lambda$-resolvable\xspace}%
{$\frac{1}{#1}$-resolvable\xspace}%
}
\usepackage{tikz}
\usetikzlibrary{automata, positioning, shapes.multipart}
\usetikzlibrary{backgrounds}
 
\theoremstyle{plain}\newtheorem*{claim}{Claim} %

\begin{document}%

\title{Resolving Nondeterminism by Chance}

\author{Soumyajit Paul\lmcsorcid{0000-0002-7233-2018}}
\author{David Purser\lmcsorcid{0000-0003-0394-1634}}
\author{Sven Schewe\lmcsorcid{0000-0002-9093-9518}}
\author{Qiyi Tang\lmcsorcid{0000-0002-9265-3011}}
\author{Patrick Totzke\lmcsorcid{0000-0001-5274-8190}}
\author{Di-De Yen\lmcsorcid{0000-0003-0045-9594}}

\address{University of Liverpool, UK}	%
\email{\{soumyajit.paul, d.purser, s.schewe, qiyi.tang, totzke, d.d.yen\}@liverpool.ac.uk}  %

\begin{abstract}
  \noindent History-deterministic automata are those in which nondeterministic choices can be correctly resolved stepwise: there is a strategy to select a continuation of a run given the next input letter so that if the overall input word admits some accepting run, then the constructed run is also accepting.

Motivated by checking qualitative properties in probabilistic verification, we consider the setting where the resolver strategy can randomise and only needs to succeed with lower-bounded probability. We study the expressiveness of such stochastically-resolvable automata as well as consider the decision questions of whether a given automaton has this property.
In particular, we show that it is undecidable to check if a given NFA is $\lambda$-stochastically resolvable. This problem is decidable for finitely-ambiguous automata.
We also present complexity upper and lower bounds for several well-studied classes of automata for which this problem remains decidable.
\end{abstract}

\maketitle

\section{Introduction}\label{sec:intro}

Many successful verification techniques rely on automata representations for specifications that capture the languages of acceptable system traces.
These automata models are typically nondeterministic: any given trace may give rise to multiple runs of the automaton, and is accepted if at least one of these runs is successful.
This enables succinct representations but is also a major source of complexity due to costly intermediate determinisation steps.

It is therefore natural to put extra constraints on the extent to which an automaton allows nondeterministic choice to avoid determinisation.
One way to do this is to bound the level of ambiguity:
an automaton is $k$-ambiguous if on every accepted word it has at most $k$ many distinct accepting runs (see e.g.~\cite{schmidt1978succinctness,OKHOTIN201215,DAVIAUD202178}).
Unambiguous automata have proven to be useful for model-checking of Markov chain models \cite{BaierK23,LPST25}.
An orthogonal restriction on nondeterminism, which has been extensively studied in recent years
(cf.~\cite{HP06,Col09,KS15,BK18,BL19,RK19,Sch20,BL21,BL22,LZ22,LPST25}), is to demand that choices can be resolved ``on-the-fly'':
An automaton is \emph{history-deterministic} (HD) if there is a strategy to select a continuation of a run given a next letter, so that if the overall word admits some accepting run then the constructed run is also accepting.
This condition is strict in the sense that the resolver strategy must guarantee to produce an accepting run if one exists. In some scenarios less strict guarantees may suffice,
for example when model-checking Markov Decision Processes \cite{HahnPSSTW20,HahnLST15,Baier2018,baier-mdp1}.
This motivates the study of automata that can be resolved in a weaker sense, namely, where the resolver strategy can randomise and is only required to succeed with a lower-bounded confidence.

A \emph{stochastic resolver} for some nondeterministic automaton $\mach{A}$
is a function that, for any finite run and input letter, gives a distribution over the possible transitions.
This produces a probabilistic automaton $\mach{P}$ that assigns a probability of acceptance to every input word and, together with a threshold $\lambda>0$, defines the language $\Lang(\mach{P}_{\geq\lambda})$ consisting of all words whose probability of acceptance is at least $\lambda$ (see, e.g.,~\cite{paz2014introduction}).
Unless stated otherwise, we will consider \emph{memoryless} stochastic resolvers, which base their decisions solely on the last state of the given run and the input letter.
Fixing a memoryless resolver turns $\mach{A}$ into a probabilistic automaton %
$\mach{P}$ over the same set of states and where transitions that appear positively in $\mach{P}$ are also contained in $\mach{A}$.
Consequently then, for every threshold $\lambda > 0$, the language $\Lang(\mach{P}_{\geq\lambda})$ is included in that of $\mach{A}$. 
We now ask which automata admit positional resolvers that also guarantee the inclusion in the other direction.

An automaton $\mach{A}$ 
called \emph{$\lambda$-memoryless stochastically resolvable} (or simply $\lambda$-resolvable)
if there exists a memoryless stochastic resolver with $\Lang(\mach{P}_{\geq\lambda}) = \Lang(\mach{A})$.
An automaton is \emph{positively memoryless stochastically resolvable} (or simply \mr) if it is $\lambda$-resolvable for some $\lambda>0$.
By varying the threshold $\lambda$, stochastic resolvability defines a spectrum where $0$-resolvable corresponds to unrestricted nondeterminism and, on finite words, $1$-resolvable coincides with history-determinism. See \cref{fig:intro} for distinguishing examples. 

\begin{figure}[t]
\begin{subfigure}[t]{0.5\textwidth}
\centering
\scalebox{0.9}{
\begin{tikzpicture}
\centering
[node distance=160pt]
\tikzstyle{state}=[draw,shape=circle,minimum size=20pt]
\tikzset{invisible/.style={minimum width=0mm,inner sep=0mm,outer sep=0mm}}

\node[state,initial, initial text=] (p) at (0,0) {$q_0$};
\node[state] (q) at (1.5,0.5) {$p$};
\node[state] (r) at (1.5,-0.5) {$q$};
\node[state,accepting] (qf) at (3,0) {$q_f$};

\path
(p)
    edge [->] node[above] {$a$} (q)
(p)
    edge [->] node[above] {$a$} (r)
(q)
    edge [->] node[above] {$b$} (qf)
(r)
    edge [->] node[above] {$c$} (qf)    

;

\end{tikzpicture}
} \caption{\nfa $\mach{A}$ is $1/2$-resolvable.}%
\label{fig:eg_nfa_a}
\end{subfigure}\hfill
\begin{subfigure}[t]{0.5\textwidth}
\centering
\scalebox{0.9}{
\begin{tikzpicture}
\centering
[node distance=160pt]
\tikzstyle{state}=[draw,shape=circle,minimum size=20pt]

\tikzset{invisible/.style={minimum width=0mm,inner sep=0mm,outer sep=0mm}}

\node[state,initial, initial text=] (q0) at (0,0) {$q_0$};
\node[state,accepting] (qf) at (2,0) {$q_f$};
\node[invisible] (r) at (1.5,-0.7) {};

\path
(q0)
    edge [->] node[above] {$b$} (qf)    
(q0)
    edge [->, loop above] node[above] {$a,b$} (q0)   

;

\end{tikzpicture}
}
\caption{\nfa $\mach{B}$ is not positively resolvable.} 
\label{fig:eg_nfa_b}
\end{subfigure}
\caption{Two unambiguous \nfa (all missing transitions implicitly go to a non-accepting sink).
The one on the left is $\lambda$-resolvable for all $\lambda\le 1/2$. The one on the right is not positively resolvable, because no matter the choice of transition probability, the probability of $b^n$ tends to zero as $n$ grows.}
\label{fig:intro}
\end{figure}

\paragraph*{Our Contributions}
We focus on automata over finite words
and consider the decision problems whether a given automaton is resolvable.
We distinguish the two variants of this problem, asking whether the automaton is positively resolvable, or whether it is $\lambda$-resolvable for a given value of $\lambda$. Our results are as follows, see also \cref{tab:lambda-SR} for a summary.
\begin{itemize}
	\item
		We introduce the quantitative notion of $\lambda$-memoryless stochastic resolvability, %
        and show $\lambda$-resolvability induces a strict hierarchy of automata with varying parameter $\lambda\in(0,1)$.

	\item We show that checking $\lambda$-resolvability is undecidable
		already for \nfa, on finite words, and therefore also for automata on infinite words regardless of the accepting condition.

	\item We complement this by showing that both positive resolvability and $\lambda$-resolvability remain decidable for 
		finitely-ambiguous automata.

	\item We present complexity upper and lower bounds for several well-studied variations of finitely-ambiguous automata (summarised in \cref{tab:lambda-SR}). 
		In particular, checking positive resolvability is \PSPACE-complete for finitely-ambiguous automata, \NL-complete for unambiguous ($1$-ambiguous) automata, and in the polynomial hierarchy for (unrestricted) unary automata. 
        Checking $\lambda$-resolvability is \PSPACE-hard even for $k$-ambiguous automata (and \coNP-hard over a unary alphabet). 
        \item We show that our decidability results for finite-ambiguous automata, and undecidability results in the general case carry over to $\omega$-regular automata.
\end{itemize}

\begin{table}[t]
	\centering
	\renewcommand{\arraystretch}{1.5} %
	\resizebox{\linewidth}{!}{%
		\begin{tabular}{|c|l|c|c|c|}
			\hline
			\multicolumn{2}{|c|}{}& \textbf{unambiguous}  & \textbf{finitely-ambiguous} & \textbf{general} \\
			\hline
			{\multirow{2}{*}{\makecell{\textbf{PR}}}} & unary     &  \NL~(Thm.~\ref{thm:1-ambiguous-non-fixed})  &
			\multicolumn{2}{c|}{\coNP-hard (Thm.~\ref{thm:unaryhardness}) $\Sigma_2^{P}$ (Thm.~\ref{thm:unarysigma2})}
			\\
			\cline{2-5}
			& non-unary & \NL-complete (Thms.~\ref{thm:ufa-NL-hard}~and~\ref{thm:1-ambiguous-non-fixed}) &  \multicolumn{1}{c|}{$\PSPACE$-complete (Thms.~\ref{thm:SR-fixed-k}~and~\ref{thm:SR-fnfa-pspapce})} & open \\
			\hline
			{\multirow{2}{*}{\makecell{\textbf{$\lambda$R}}}} & unary     &  \PTIME~(Thm.~\ref{thm:ufa-lambda-resolvable}) &
			$\coNP$-hard (Thm.~\ref{thm:unaryhardness}) decidable  (Thm.~\ref{thm:constant-fnfa-lambda-resolvable}) & open \\
			\cline{2-2}\cline{3-5}
			& non-unary & \NL-hard (Thm.~\ref{thm:ufa-NL-hard}) \PTIME~(Thm.~\ref{thm:ufa-lambda-resolvable}) & $\PSPACE$-hard (Thm.~\ref{thm:SR-fixed-k}) decidable (Thm.~\ref{thm:constant-fnfa-lambda-resolvable}) & undecidable (Thm.~\ref{thm:undecidable}) \\
			\hline
		\end{tabular}
	}~\\
	\caption{Deciding positive resolvability (PR) and $\lambda$-resolvability ($\lambda$R) of NFAs where $\lambda \in (0,1)$.}
	\label{tab:lambda-SR}
\end{table}

\paragraph*{Related Work}
The notion of history-determinism was introduced independently, with slightly
different definitions, by Henzinger and Piterman~\cite{HP06} for solving games without determinisation, by Colcombet~\cite{Col09} for cost-functions, and by Kupferman, Safra, and Vardi~\cite{KSV06} for recognising derived tree languages of word automata.
For $\omega$-regular automata,
these variations all coincide~\cite{BL19}.
For coB\"uchi-recognisable languages, history-deterministic automata can be
exponentially more succinct than any equivalent deterministic ones~\cite{KS15}.
Checking whether a given automaton is HD is decidable in polynomial time
for B\"uchi and coB\"uchi automata \cite{BK18,KS15} and more generally also for parity automata of any fixed parity index~\cite{LP2025}.
History-determinism has been studied for richer automata models, such as
pushdown automata~\cite{LZ22,GJLZ21} timed automata \cite{BHLST2023,BHLST2022,HLT2022} and quantitative automata~\cite{BL21,BL22}.

Recently, and independently of us, Henzinger, Prakash and Thejaswini \cite{henzingerEtAl2025} introduced classes of \emph{stochastically resolvable} automata. %
The main difference of their work compared to ours is that they study the qualitative setting, where resolvers are required to succeed \emph{almost-surely}, with probability one.
In the terminology introduced above these are $1$-resolvable automata.
They study such automata on infinite words
and compare their expressiveness, relative succinctness, against history-deterministic and semantically deterministic automata. Over finite words, these notions coincide: an NFA is $1$-stochastically resolvable iff it is $1$-memoryless stochastically resolvable iff it is semantically deterministic iff it is history-deterministic.
They also consider the complexity of determining whether an automaton is $1$-resolvable and establish that this is in polynomial time for safety automata, and \PSPACE-complete for reachability and weak automata, and remains open for more general classes.

For a fixed probabilistic automaton, deciding whether a given $\lambda$ is a lower bound corresponds to the undecidable emptiness problem~\cite{paz2014introduction}, while asking if such a $\lambda$ exists relates to the undecidable zero-isolation problem~\cite{Gimbert2010}. These problems become more tractable under restrictions such as bounded ambiguity~\cite{CzerwinskiLMPW22,DAVIAUD202178,FijalkowGKO15,FIJALKOW2022}. In the unary case, checking if $\lambda$ is a lower bound is equivalent to the long-standing open positivity problem for linear recurrence sequences~\cite{OuaknineW14} and is also hard for Markov chains~\cite{Vahanwala24}. In contrast, our setting allows the probabilistic automaton to vary, asking whether some choice of probabilities satisfies the desired properties.
A similar construction appears for weighted automata, where \emph{stochastization}
assigns probabilities to the nondeterministic choices and one asks whether the resulting
probabilistic automaton approximates the original value~\cite{AvniK15}.
Their stochastizations are memoryless resolvers in our sense.

The unary case can be seen as a synthesis problem on parametric Markov chains. Consider the distribution transformer perspective~\cite{AghamovBKNOPV25,Baier0JKLOPW022}, where a Markov chain induces a sequence of distributions over states. Then the $\lambda$-resolvability problem asks, for universal NFA, whether there exists a parameter assignment such that the sequence of distributions is ``globally'' in the semi-algebraic set $S =\{x \in [0,1]^Q | \sum_{q\in Q} x_q = 1 \text{ and } \sum_{q \in F} x_q \ge \lambda\}$ (distributions with probability mass at least $\lambda$ in accepting states). When the NFA is non-universal the ``globally'' condition must be adjusted to the eventually periodic pattern of the NFA. Existing work characterises, up to a set of measure zero, parameters that satisfy prefix-independent properties or ``eventually'' properties such as Finally Globally in $S$~\cite{Baier0JKLOPW022}. For general parametric questions, there is always the special case where the parameters are redundant, which encodes the aforementioned positivity problem~\cite{Vahanwala24}. This special case does not necessarily arise in our setting, as the problem has \emph{all} transitions parameterised independently (only the structure is fixed), which leaves hope for $\lambda$-resolvability.

\section{Preliminaries}
In this paper, we use $\nats$, $\ints$, $\positives$, $\rats$, $\positiverats$, and $\reals$ to denote the sets of non-negative integers, integers, positive integers, rational numbers, positive rational numbers, and real numbers, respectively.
For $S \in \{\ints, \rats, \reals\}$ and two numbers $i, j \in S$ with $i < j$, the notation $[i,j]_{S}$ (resp., $(i,j)_{S}$) denotes the set $\{d \in S \mid i \leq d \leq j\}$ (resp., $\{d\in S\mid i< d <j\}$). The half-open intervals $(i, j]_S$ and $[i, j)_S$ are defined analogously. We also write $[j]_{S}$ to denote $[1,j]_{S}$.  The subscript $S$ is omitted if it is clear from the context. 

\subparagraph{Nondeterministic Finite Automata}
A {\em nondeterministic finite automaton} (\nfa)  $\mach{A}$ consists of
a finite alphabet $\Sigma$,
a finite set of states $Q$,
an initial state $q_0$,
a set of accepting states $F\subseteq Q$, and
a set of transitions $\Delta\subseteq Q\times \Sigma \times Q$.
We also use $p \xrightarrow{\sigma} q$ to signify that $(p,\sigma, q)\in\Delta$.
$\mach{A}$ is \emph{unary} if $|\Sigma| = 1$; a \emph{deterministic finite automaton} (\dfa) if, for every $(p, \sigma) \in Q \times \Sigma$, there exists at most one state $q \in Q$ such that $p \xrightarrow{\sigma} q$; and \emph{complete} if, for every $(p, \sigma) \in Q \times \Sigma$, there exists some $q \in Q$ such that $p \xrightarrow{\sigma} q$.

Given a word $w = \sigma_1 \dots \sigma_n \in \Sigma^*$, a \emph{run} $\pi$ of $\mach{A}$ on $w$ is a sequence of transitions $\tau_1 \dots \tau_n$, where each transition $\tau_i = (p_i, \sigma_i, q_i) \in \Delta$ for $i \in [1, n]$, and $q_i = p_{i+1}$ for $i \in [1, n-1]$. The run $\pi$ is \emph{accepting} if $p_1 = q_0$ and $q_n \in F$. 
A word is \emph{accepted} by $\mach{A}$ if there exists an accepting run on it. We denote by $\machlang{\mach{A}}$ the set of words accepted by $\mach{A}$. Additionally, we use $\accsub{\mach{A}}{w}$ to denote the set of all accepting runs of $\mach{A}$ on $w$. We omit the subscript $\mach{A}$ when it is clear from the context.  
An \nfa $\mach{A}$ is {\em \kamb{k}} if, for every word $w \in \Sigma^*$, it holds that $|\acc{w}| \leq k$. It is \unamb when $k=1$. Moreover, $\mach{A}$ is {\em \finamb} if it is \kamb{k} for some $k \in \positives$. We use \unfa and \fnfa to refer to unambiguous and finitely-ambiguous \nfa, respectively.

For $q \in Q$, let $\A_q$ denote the automaton $\A$ with $q$ as its initial state, and for $S \subseteq \Trans$, let $\A_S$ denote the automaton obtained from $\A$ by restricting its transitions to $S$. For $\tau \in S$, we often say $\tau$ is nondeterministic in $S$ to mean $\tau$ is nondeterministic in $\A_S$, that is, for $\tau = (p,\sigma,q)$ there does not exist $(p,\sigma,q')\in S$ with $q\ne q'$.
$\Trans$ induces a transition function $\delta_{\A} : 2^Q \times \Sigma^* \mapsto 2^Q$, where $\delta_{\mach{A}}(Q',w)$ is the set of all states where runs on word $w$ can end up in when starting from any state $q \in Q'$. Formally, for all $Q' \incl Q$, $\delta_{\mach{A}} (Q',\emptyword) = Q'$, and for all $ w \in \Sigma^*$ and $ \sigma \in \Sigma$, $\delta_{\mach{A}} (Q',w\cdot\sigma) = \{q \mid (p,\sigma,q) \in \Trans \text{ and } p \in \delta_{\mach{A}}(Q',w)\}$. When $Q' = \{p\}$, we often abuse notation and write $\delta_{\mach{A}}(p, w)$ instead of $\delta_{\mach{A}}(\{p\}, w)$, and we omit the subscript $\mach{A}$ when it is clear from the context. 
An \nfa is \emph{trim} if every state $q \in Q$ is reachable from the initial state, and can reach an accepting state, i.e., $q \in \delta(q_0,w)$ and $\delta(q,w')\cap F\neq \emptyset$ for some words $w,w'$. 

\subparagraph{Probabilistic Finite Automata}
A {\em probabilistic finite automaton} ({\pfa}) $\mach{P}$ is an extension of nondeterministic finite automaton that assigns an acceptance probability to each word.
Specifically, $\mach{P}$ is an \nfa augmented with an assignment  $\Theta: \Delta \rightarrow [0,1]_{\reals}$, where for every $(p,\sigma) \in Q\times \Sigma$, 
$\sum_{\tau=(p,\sigma,q) \in \Delta}\Theta(\tau) = 1$.
Accordingly, $\mach{P}$ is a 6-tuple of the form $(\Sigma,Q,q_0,\Delta,F, \Theta)$.
Given a transition $\tau = (p, \sigma, q)$, the value $\Theta(\tau)$ represents the probability of transitioning from state $p$ to state $q$ upon reading the symbol $\sigma$.  
With a slight abuse of language, we also call $(p, \sigma, q, d)$ a transition of $\mach{P}$, where $d = \Theta(p, \sigma, q)$.
We say that $\mach{P}$ is {\em based on} the \nfa $\mach{A}=(\Sigma,Q,q_0,\Delta,F)$ and call this its \emph{underlying} \nfa.\footnote{The condition $\sum_{\tau = (p, \sigma, q) \in \Delta} \Theta(\tau) = 1$ may not hold for some pairs $(p, \sigma)$ if the underlying \nfa is not complete. Any \nfa can be made complete by introducing a non-accepting sink state. Thus, throughout this paper, we assume that the underlying \nfa of every \pfa is implicitly made complete.}

We use $\mach{P}(w)$ to denote the probability that $\mach{P}$ accepts word $w$ where: 
\[
\mach{P}(w) := \sum_{\pi = \tau_1 \dots \tau_n \in \acc{w}} \prod_{i=1}^{n} \Theta(\tau_i).
\]
For every threshold $\lambda \in [0,1]_{\reals}$, let $\Lang(\mach{P}_{\geq \lambda})$ denote the set of all words accepted with probability at least $\lambda$, that is:
$
\Lang(\mach{P}_{\geq \lambda}) := \{w \in \Sigma^* \mid \mach{P}(w) \geq \lambda\}.
$

We also use $\mach{P}_{\geq \lambda}$ to refer to the \pfa $\mach{P}$ with a given threshold $\lambda$. Similarly, we define $\mach{P}_{\leq \lambda}$, $\mach{P}_{> \lambda}$, and $\mach{P}_{< \lambda}$.
A \pfa $\mach{P}$ is called \emph{simple} if, for every transition $(p, a, q, d)$ of $\mach{P}$, the probability $d$ belongs to $\{0, \frac{1}{2}, 1\}$. Since transitions with $d = 0$ do not affect the language accepted by a \pfa, we assume $d \neq 0$ for all transitions in the remainder of this paper.

\subparagraph{Stochastic Resolvers}
Given an \nfa $\A$, a (memoryless) \emph{stochastic resolver} $\mach{R}$ for $\A$ resolves nondeterminism during the run on a word on-the-fly, randomly, and positionally. Formally, $\mach{R}: \Delta \mapsto [0,1]$, such that  
$\sum_{\tau=(p,\sigma,q)  \in \Delta} \mach{R}(\tau) = 1$ holds for every pair $(p,\sigma)\in Q\times\Sigma$.
A resolver $\mach{R}$ for $\A$ turns it into a \pfa $\mach{P}^{\A}_{\mach{R}}$.
We call  $\{ \tau \in \Trans \mid \mach{R}(\tau) > 0\}$ the \emph{support} of $\mach{R}$. 

Given a threshold $\lambda \in [0,1]$, we say that $\mach{A}$ is {\em $\lambda$-memoryless stochastically resolvable} (or simply {\em \lmr{}}) if there exists a resolver $\mach{R}$ for $\A$ such that 
$\Lang(\mach{P'}_{\geq \lambda}) = \Lang(\mach{A})$, where $\mach{P'}=\mach{P}^{\A}_{\mach{R}}$.
An automaton $\mach{A}$ is {\em positively memoryless stochastically resolvable} (or simply {\em \mr}) if it is \lmr{} for some $\lambda \in (0,1]$.

\section{\texorpdfstring{$\lambda$}{Lambda}-Resolvability}
\label{sec:lambda}

In this section we demonstrate the importance of the threshold $\lambda$: we first show that the choice of this threshold defines a spectrum of automata classes in terms of their resolvability. We then turn to the main negative result, that checking $\lambda$-resolvability is undecidable.

\begin{thm}\label{thm:lambda-spectrum}
For every $\lmb \in (0,1)_{\rats}$ there exists a unary NFA $\A_{\lmb}$ such that $\A_{\lmb}$ is $\lmb$-resolvable but not $(\lmb + \epsilon)$-resolvable for any $\epsilon>0$.
\end{thm}

\begin{figure}[t]
\centering
\scalebox{0.85}{
\begin{tikzpicture}
\tikzset{
yscale=1,>=latex',shorten >=1pt,node distance=2.5cm,
every state/.style={inner sep =.04cm,minimum size=0.1cm},
accstate/.style={ double,draw=black, circle,line width=0.7pt,double distance=1.5pt,inner sep =.06cm,outer sep=1pt, minimum size=1},
on grid,auto,initial text = {}	
}

\node[state,initial left] 	   (s)  {\Large$q_0$};
\node[accstate, right of=s]  (q2) {\textcolor{teal}{$q^2_{\{1,2\}}$}};
\node[accstate, above =1.5 of q2]      (q1) {\textcolor{teal}{$q^1_{\{1,2\}}$}};
\node[state, below=1.5 of q2] (q3) {$q^3_{\{1,2\}}$};

\node[state, right of= q1]  (q11) {$q^1_{\{2,3\}}$};
\node[accstate, right of= q2]  (q22) {\textcolor{blue}{$q^2_{\{2,3\}}$}};
\node[accstate, right of= q3]  (q33) {\textcolor{blue}{$q^3_{\{2,3\}}$}};

\node[accstate, right of= q11]  (q111) {\textcolor{red}{$q^1_{\{3,1\}}$}};
\node[state, right of= q22]  (q222) {$q^2_{\{3,1\}}$};
\node[accstate, right of= q33]  (q333) {\textcolor{red}{$q^3_{\{3,1\}}$}};

\path[->] (s) edge  [above]node {\color{black} $a$} (q1);
\path[->] (s) edge  [above]node {\color{black} $a$} (q2);
\path[->] (s) edge  [above]node {\color{black} $a$} (q3);

\path[->] (q1) edge  [above]node {\color{black} $a$} (q11);
\path[->] (q2) edge  [above]node {\color{black} $a$} (q22);
\path[->] (q3) edge  [above]node {\color{black} $a$} (q33);

\path[->] (q11) edge  [above]node {\color{black} $a$} (q111);
\path[->] (q22) edge  [above]node {\color{black} $a$} (q222);
\path[->] (q33) edge  [above]node {\color{black} $a$} (q333);

\end{tikzpicture}
}
\caption{The automaton $\A_{\lambda}$ for $\lambda = \frac{2}{3}$, in the proof of \cref{thm:lambda-spectrum}.
}%
\label{fig:lambda-sr-hierrarchy}
 \end{figure}
\begin{proof}
Let $\lmb = \frac{m}{n}$ where $m,n \in \positives$ and $m < n$. The NFA $\A_{\lmb}$ over unary alphabet $\{a\}$ is as follows. In the initial state $q_0$ upon reading the letter $a$, $\A_{\lmb}$ nondeterministically goes to one of $n$ branches, each with $n\choose m$ states. Consider an arbitrary ordering of all $m$ sized subsets of $[n]$, $S_1,\dots, S_{n\choose m}$. For any branch $j \in [n]$ and $i \in  [{n\choose m}] $, the $i$-th state on the $j$-th branch - $q^j_{S_i}$,  is accepting iff $j \in S_i$. Refer to \cref{fig:lambda-sr-hierrarchy} for an example with $\lambda=\frac{2}{3}$.
The automaton $\A_{\lmb}$ is $m$-ambiguous, where each word in $\{a^i| i \in [{n\choose m}]\}$ has exactly $m$ accepting runs. Moreover, $\A_{\lmb}$ is $\frac{m}{n}$-resolvable, since the uniform resolver that picks any transition at $q_0$ with probability $\frac{1}{n}$,  accepts every word with probability exactly $\frac{m}{n}$.
For any other resolver $\mach{R}$, consider the $m$ branches $i_1, \dots, i_m$, with the lowest assigned probabilities.
Let $\ell$ be the unique index such that $S_{\ell} = \{i_1, \dots, i_m\}$. Then for the word $a^{\ell}$, we have $\mach{P}^{\A_{\lambda}}_{\mach{R}}(a^{\ell}) < \frac{m}{n}$. Hence, $\A_{\lmb}$ is not $(\lmb + \epsilon)$-resolvable for any $\epsilon>0$.
\end{proof}

\begin{figure}[t]
\centering
\scalebox{0.85}{
\begin{tikzpicture}[node distance=50pt]
\tikzstyle{state}=[draw,rectangle, rounded corners,minimum size=25pt]
\tikzstyle{fstate}=[rectangle, rounded corners,minimum size=25pt]

\node[state] (P) at (0,2.7) {$\mach{P}$};
\node[fstate] (FP) at (0,2.8){};
\node at (0,2.1) {over $\Sigma$};
\node[fstate] (FFP) at (0,2.6){};
\node[state] (A) at (5,2.7) {$\mach{A}$};
\node[fstate] (FA) at (5,2.8) {};
\node[fstate] (FFA) at (5,2.6) {};
\node at (5.6,2.1) {over $\Sigma$};
\node[state] (AP) at (5,0) {$\mach{A'}$};
\node at (0,-0.6) {over $\Sigma'$};
\node[state] (PP) at (0,0) {$\mach{P'}$}; 
\node at (5,-0.6) {over $\Sigma'$};

\path 
(FP) edge [->] node[above] {based on} (FA)
(FFA) edge [->] node[below] {underlying} (FFP)
(PP) edge [->] node[above] {unique simple \pfa} (AP)
(PP) edge [->] node[below] {based on} (AP)
(AP) edge [->] node[left] {\rotatebox{-90}{extends}} (A)
    ;
\end{tikzpicture}
}
\caption{The construction for \cref{thm:undecidable}.}
\label{fig:flow}
\end{figure}

\label{sec:undecidable}
\begin{thm}
\label{thm:undecidable}
    The $\lambda$-resolvability problem is undecidable for \nfa.
\end{thm}

Our approach relies on a reduction from the undecidable universality  problem for simple \pfa{s}~\cite{paz2014introduction,Gimbert2010} (\cref{lem:simple_universal} below) to the $\lambda$-resolvability problem for an {\nfa}.  
Consider a simple {\pfa} $\mach{P}$ over the alphabet $\Sigma$, whose underlying {\nfa} is $\mach{A}$.
For a simple \pfa, the structure of its underlying \nfa uniquely determines the probabilities on its transitions: for each state $q$ and letter $\sigma$, there are either one or two outgoing transitions from $q$ labelled with $\sigma$; if there is one transition it must have probability $1$ and if there are two, each must have probability $\frac{1}{2}$.

Intuitively, we want the resolver (if it exists) for the underlying automaton $\mach{A}$ of $\mach{P}$ to always induce a simple \pfa. So, $\mach{A}$ is \lmr{} if and only if $\mach{P}_{\geq \lambda}$ and $\mach{A}$ recognise the same language. In principle, a resolver for an \nfa does not necessarily enforce that two outgoing transitions both have probability $\frac{1}{2}$ and may use any arbitrary split. To enforce an equal split, we construct an \nfa $\mach{A'}$ over an extended alphabet $\Sigma'$, such that $\mach{A'}$ is a \emph{super-automaton} of $\mach{A}$. By this, we mean that the alphabet, state set, set of accepting states, and transition set of $\mach{A}$ are all subsets of those of $\mach{A'}$, and we refer to $\mach{A}$ as a sub-automaton of $\mach{A'}$.
The constructed \nfa $\mach{A'}$ is designed so that it admits only one possible $\frac{1}{4}$-resolver, and that the resolver matches the underlying probabilities of the simple \pfa $\mach{P}$ on the part of the automaton of $\mach{A'}$ corresponding to $\mach{A}$. These relations are described in \cref{fig:flow}.

In \cite{Gimbert2010}, it is shown that determining whether $\Lang(\mach{P}_{>\frac{1}{2}})$ is empty for a given simple \pfa $\mach{P}$ is undecidable.  
For a given \pfa $\mach{P} = (\Sigma, Q, q_0, F, \Delta, \Theta)$, define its complement \pfa as $\mach{P}^{\text{comp}} = (\Sigma, Q, q_0, F_{\text{comp}}, \Delta, \Theta)$, where $F_{\text{comp}} = Q \setminus F$.  
By the definition of \pfa{s}, we have  
$\mach{P}(w) + \mach{P}^{\text{comp}}(w) = 1$, for all $w \in \Sigma^*$.
Thus, $\Lang(\mach{P}_{>\frac{1}{2}})$ is empty if and only if $\Lang(\mach{P}^{\text{comp}}_{\geq\frac{1}{2}})$ is universal. This leads to the following lemma:

\begin{lem}\label{lem:simple_universal}
It is undecidable whether, for a given simple \pfa $\mach{P}$, $\Lang(\mach{P}_{\geq\frac{1}{2}})$ is universal.
\end{lem}

To prove \cref{thm:undecidable}, we show that the universality problem of $\mach{P}_{\geq\frac{1}{2}}$ for a given simple \pfa $\mach{P}$ can be effectively reduced to the \lmr{4} problem for an \nfa $\mach{A'}$ derived from the underlying \nfa $\mach{A}$ of $\mach{P}$.
To establish this uniqueness property, we need the following lemma:

\begin{lem}\label{prop:uniqueness}
Given an alphabet $\Sigma$ and an \nfa $\mach{A}$ over $\Sigma$, if for every state $p$ of $\mach{A}$ and every $\sigma \in \Sigma$, there are at most two distinct transitions of the form $(p, \sigma, q)$ for some state $q$, then the simple \pfa $\mach{P}$ based on $\mach{A}$ is unique.
\end{lem}
\begin{proof}
Let $\mach{P}=(\Sigma,Q,q_0,\Delta,F,\Theta)$.
If there is only one outgoing transition from $p$ on $\sigma$, its probability must be 1.
Suppose $(p, \sigma, q)$ and $(p, \sigma, q')$ are two distinct transitions of $\mach{A}$. Since $\mach{P}$ is based on $\mach{A}$, the transitions $(p, \sigma, q, d)$ and $(p, \sigma, q', d')$ exist in $\mach{P}$ for some probabilities $d$ and $d'$. 
By assumption, every transition $(p, \sigma, q, d)$ in $\mach{P}$ has a nonzero probability, i.e., $d \neq 0$, and for every pair $(p,\sigma)\in Q\times \Sigma$, we have $\sum_{\tau=(p,\sigma,q)}\Theta(\tau)=1$. Therefore, we have $d + d' = 1$. Since $\mach{P}$ is simple, it follows that $d, d' \in \{\frac{1}{2}, 1\}$, which implies $d = d' = \frac{1}{2}$.
\end{proof}

Now we are ready to give the details of the proof of \cref{thm:undecidable}.
\begin{proof}
Without loss of generality, we fix $\Sigma = \{a, b\}$. If $\mach{A}$ is not universal, then $\mach{P}_{\geq\frac{1}{2}}$ is also not universal. Since the universality problem for {\nfa} is decidable, we may further assume that $\machlang{\mach{A}} = \Sigma^*$.
To derive the undecidability result, we construct the super-automaton $\mach{A'}$ of $\mach{A}$ over the alphabet  $\Sigma' = \Sigma \cup Q \cup \Delta \cup \Sigma_{\$}$, where $\Sigma_{\$} = \{\$_0, \dots, \$_6\}$. The construction satisfies the following three properties:

\begin{enumerate}[C.1]
    \item\label{prop:c1} For every state $p$ of $\mach{A'}$ and every $\sigma \in \Sigma'$, $\mach{A'}$ has at most two distinct transitions of the form $(p,\sigma,q)$ for some state $q$.
    \item\label{prop:c2} If $\mach{A'}$ is \lmr{4}, then its \lmr{4} solution must be a simple \pfa.
    \item\label{prop:c3} Let $\mach{P'}$ be the simple \pfa based on $\mach{A}'$. The language $\machlang{\mach{A'}}$ is the disjoint union of two languages, $\Lang_{\text{ext}}$ and $\Lang_{\text{ind}}$, satisfying:
    \begin{enumerate}[a]
        \item\label{prop:c3a} For every $w\in \Lang_{\text{ind}}$, $\mach{P'}(w) \geq \frac{1}{4}$.
        \item\label{prop:c3b} For every $w'' \in \Lang_{\text{ext}}$, where $\Lang_{\text{ext}} = S \cdot \Sigma^*$ for some singleton set $S = \{w'\}$ over $\Sigma'\setminus\Sigma$ and $w'' = w' w$ with $w \in \Sigma^*$, we have $\mach{P'}(w'') = \frac{1}{2} \cdot \mach{P}(w)$.
    \end{enumerate}
\end{enumerate}

By \cref{prop:uniqueness}, the simple \pfa $\mach{P'}$ based on $\mach{A'}$ is unique if $\mach{A'}$ satisfies \cref{prop:c1}. Therefore, if $\mach{A'}$ further satisfies \cref{prop:c2}, then $\mach{P'}$ is the unique \lmr{4} solution to $\mach{A'}$ provided that $\mach{A'}$ is \lmr{4}.
By \cref{prop:c3}.\cref{prop:c3a}, every word in $\Lang_{\text{ind}}$ is accepted by $\mach{P}'$ with probability at least $\frac{1}{4}$. By \cref{prop:c3}.\cref{prop:c3b}, $\mach{P'}(w'w)\geq \frac{1}{4}$ if and only if $\mach{P}(w)\geq \frac{1}{2}$, for every $w\in \Sigma^*$. Hence, $\mach{P'}$ is a $\frac{1}{4}$-resolvable solution to $\mach{A'}$ if and only if $\mach{P}_{\geq\frac{1}{2}}$ is universal.
Accordingly, constructing $\mach{A'}$ to satisfy \cref{prop:c1}--\cref{prop:c3} leads to the following result:

Let $\mach{A}$ be the underlying \nfa of a given simple \pfa $\mach{P}$, and let $\mach{A'}$ be the super-automaton of $\mach{A}$ satisfying properties \cref{prop:c1}--\cref{prop:c3}. Then $\mach{P}_{\geq\frac{1}{2}}$ is universal iff $\mach{A'}$ is \lmr{4}.

The automaton $\mach{A'}$ is over $\Sigma'$. We also treat states $Q$ and transitions $\Delta$ of $\mach{A}$ as symbols. Without loss of generality, we assume that $\Sigma, Q, \Delta$, and $\Sigma_{\$}$ are pairwise disjoint. The language $\Lang(\mach{A'})$ consists of two main parts (recalling that $\Sigma^* = \machlang{\mach{A}}$):
\begin{itemize}
    \item $\Lang_{\text{ext}} = \{\$_0\} \cdot \{q_0\} \cdot \Sigma^*$
    \item $\Lang_{\text{ind}}$, which is the union of:
    \begin{itemize}
        \item $\Lang_{\$} = \{\$_0\$_1\$_i, \$_0\$_2\$_j \mid i \in \{3,4\}, j \in \{5,6\}\}$.
        \item $\Lang_{\text{extra}} = \{\$_0\} \cdot (Q \setminus \{q_0\}) \cdot \Sigma^* \cup \{\$_0\} \cdot Q \cdot (\Sigma^* \setminus \Sigma) \cdot \Delta$.
        \item $\Lang_{\text{nondet}} = \{\$_0\} \cdot \{p\sigma\tau \in Q\cdot\Sigma\cdot\Delta \mid \tau=(p,\sigma,s) \in \Delta \text{ is a nondeterministic transition}\}$.
    \end{itemize}
\end{itemize}

\begin{figure}[t]
\centering
\begin{tikzpicture}
\centering
[node distance=160pt]
\tikzstyle{state}=[draw,shape=circle,minimum size=20pt]

\tikzset{invisible/.style={minimum width=0mm,inner sep=0mm,outer sep=0mm},initial text = {}}

\node[state,initial] (qp0) at (-3,0) {$q'_0$};
\node[state] (x) at (-1.5,1) {$x$};
\node[state] (y) at (-1.5,-1) {$y$};
\node[state] (q1) at (1,1.5) {$q_3$};
\node[state] (q2) at (1,0.5) {$q_4$};
\node[state] (q3) at (1,-0.5) {$q_5$};
\node[state] (q4) at (1,-1.5) {$q_6$};
\node[state,accepting] (qpf) at (4,0) {$q'_f$};

\path
(qp0)
    edge [->] node[above] {$\$_0$} (x)
(qp0)
    edge [->] node[above] {$\$_0$} (y)
(x)
    edge [->] node[above] {$\$_1$} (q1)
(x)
    edge [->] node[above] {$\$_1$} (q2)
(y)
    edge [->] node[above] {$\$_2$} (q3)
(y)
    edge [->] node[above] {$\$_2$} (q4)
(q1)
    edge [->] node[above] {$\$_3$} (qpf)
(q2)
    edge [->] node[above] {$\$_4$} (qpf)
(q3)
    edge [->] node[above] {$\$_5$} (qpf)
(q4)
    edge [->] node[above] {$\$_6$} (qpf)
    ;

\end{tikzpicture}
 
\caption{The \nfa $\mach{A}_{\$}$ corresponds to $\Lang_{\$}$.}
\label{fig:L_dollar}
\end{figure} 
\begin{figure}[t]
\centering
\scalebox{0.85}{
\begin{tikzpicture}
\centering
[node distance=160pt]
\tikzstyle{state}=[draw,shape=circle,minimum size=25pt]

\tikzset{invisible/.style={minimum width=0mm,inner sep=0mm,outer sep=0mm},initial text = {}}

\node[state,initial] (qp0) at (-4,0) {$q'_0$};
\node[state] (x) at (-2,0) {$x$};
\node[state] (qppf) at (1,0) {};
\node[state] (qpppf) at (3.5,0) {};
\node[state] (qppppf) at (6,0) {};

\node[state,accepting] (xp) at (-1.5,-3) {};
\node[state, accepting] (xpp) at (1,-3) {};
\node[state, accepting] (xppp) at (3.5,-3) {};
\node[state,accepting] (qf) at (6,-3) {};

\path
(qp0)
    edge [->] node[above, sloped] {$\$_0$} (x)
(x)
    edge [->] node[above, sloped] {$q_0$} (qppf)
(qppf)
    edge [->] node[above, sloped] {$\sigma\in\Sigma$} (qpppf)
(qpppf)
    edge [->] node[above, sloped] {$\sigma\in\Sigma$} (qppppf)
(qppppf)
    edge [loop right] node[right] {$\sigma\in\Sigma$} (qppppf)

(x)
    edge [->] node[below, sloped] {$p\in Q\setminus\{q_0\}$} (xp)
(qppf)
    edge [->] node[above, sloped] {$\sigma\in\Delta$} (qf)
(qppppf)
    edge [->] node[below, sloped] {$\sigma\in\Delta$} (qf)
(xp)
    edge [->] node[above, sloped] {$\sigma\in\Sigma$} (xpp)
(xpp)
    edge [->] node[above, sloped] {$\sigma\in\Sigma$} (xppp)
(xppp)
    edge [->] node[above, sloped] {$\sigma\in\Delta$} (qf)
(xp)
    edge [->, bend right] node[above, sloped] {$\sigma\in\Delta$} (qf)
(xppp)
    edge [loop above] node[left] {$\sigma\in\Sigma$} (xppp)
;

\end{tikzpicture}
 }
\caption{The \nfa $\mach{A}_{\text{extra}}$ corresponds to $\Lang_{\text{extra}}$.}
\label{fig:L_extra}
\end{figure}
 
Each of the languages $\Lang_{\text{ext}}$, $\Lang_{\$}$, $\Lang_{\text{extra}}$, and $\Lang_{\text{nondet}}$ corresponds to a sub-automaton of $\mach{A'}$, and the \nfa they correspond to share some states and transitions. Note that while each of these disjoint languages refers to a sub-automaton of $\mach{A'}$, this is not a general property of \nfa.  
The details of these sub-automata are as follows:
\begin{itemize}
    \item $\Lang_{\$}$ is recognized by the \nfa $\mach{A}_{\$}$ in \cref{fig:L_dollar}.
    \item $\Lang_{\text{extra}}$ is recognized by the \nfa $\mach{A}_{\text{extra}}$ in \cref{fig:L_extra}, where $\mach{A}_{\text{extra}}$ and $\mach{A}_{\$}$ share the states $q'_0$ and $x$, as well as the transition $(q'_0,\$_0,x)$.
    \item $\Lang_{\text{ext}}$ is recognized by the super-automaton $\mach{A}_{\text{ext}}$ of $\mach{A}$, where $\mach{A}_{\text{ext}}$ is obtained from $\mach{A}$ by replacing the initial state $q_0$ with $q'_0$, adding two transitions $(q'_0,\$_0,y)$ and $(y,q_0,q_0)$, and introducing two new states $q'_0, y$ along with the symbols $\$_0$ and $q_0$. The transition $(q'_0,\$_0,y)$ also appears in $\mach{A}_{\$}$.
    Note that all states, transitions, and accepting states of $\mach{A}$ are preserved in $\mach{A}_{\text{ext}}$, and hence also in $\mach{A'}$.

    \item $\Lang_{\text{nondet}}$ is recognized by the super-automaton $\mach{A}_{\text{nondet}}$ of $\mach{A}$, where $\mach{A}_{\text{nondet}}$ is obtained from $\mach{A}$ by replacing the initial state $q_0$ with $q'_0$ and adding the transition $(q'_0,\$_0,y)$, which is shared by both $\mach{A}_{\text{ext}}$ and $\mach{A}_{\$}$. Additionally, as shown in \cref{fig:nondet_structure}, for every $p \in Q$ and $\sigma\in\Sigma$, if $e_1=(p,\sigma,q)$ and $e_2=(p,\sigma,r)$ are two distinct transitions of $\mach{A}$, then two new transitions $(q,e_1,q'_f)$ and $(r,e_2,q'_f)$ are added, leading to a newly introduced accepting state $q'_f$.
\end{itemize}

\begin{figure}[t]
\centering
\begin{tikzpicture}
\centering
[node distance=160pt]
\tikzstyle{state}=[draw,shape=circle,minimum size=20pt]

\tikzset{invisible/.style={minimum width=0mm,inner sep=0mm,outer sep=0mm},initial text = {}}

\node[state,initial] (qp0) at (-3,0) {$q'_0$};
\node[state] (q0) at (-1.5,0) {$q_0$};
\node[state] (p) at (0,0) {$p$};
\node[state] (q) at (1.5,1) {$q$};
\node[state] (r) at (1.5,-1) {$r$};
\node[state,accepting] (qf) at (3,0) {$q'_f$};

\path
(qp0)
    edge [->] node[above] {$\$_0$} (q0)
(q0)
    edge [->] node[above] {$q_0$} (p)
(p)
    edge [->] node[above] {$\sigma$} (q)
(p)
    edge [->] node[above] {$\sigma$} (r)
(q)
    edge [->] node[above] {$e_1$} (qf)
(r)
    edge [->] node[above] {$e_2$} (qf)    

;

\end{tikzpicture}

\caption{Nondeterministic outgoing transitions from $p$.}
\label{fig:nondet_structure}
\end{figure}

According to their definitions, these automata all satisfy property \cref{prop:c1}. The reasons are as follows:
\begin{itemize}
    \item The only nondeterministic states in $\mach{A}_{\$}$ are $q'_0, x, y$, from which there are exactly two transitions.
    \item $\mach{A}_{\text{extra}}$ is deterministic.
    \item All the nondeterministic states of $\mach{A}_{\text{ext}}$ and $\mach{A}_{\text{nondet}}$ are states of $\mach{A}$. Since $\mach{A}$ is the underlying automaton of the simple \pfa $\mach{P}$, it satisfies \cref{prop:c1}, and so do $\mach{A}_{\text{ext}}$ and $\mach{A}_{\text{nondet}}$.
\end{itemize}
Therefore, in the remainder of the proof, we only need to verify that $\mach{A'}$ satisfies properties \cref{prop:c2} and \cref{prop:c3}.

\subsection*{Property \texorpdfstring{\cref{prop:c2}}{\getrefnumber{prop:c2}}}
First, consider $\mach{A}_{\$}$. Let $\mach{P}_{\$}$ be a \pfa based on $\mach{A}_{\$}$, where the transitions of $\mach{P}_{\$}$ are  
$(q'_0,\$_0,x,d_1)$, $(q'_0,\$_0,y,d_2)$, $(x,\$_1,q_3,d_3)$, $(x,\$_1,q_4,d_4)$, $(y,\$_2,q_5,d_5)$, $(y,\$_2,q_6,d_6)$, and $(q_i,\$_i,q'_f,1)$ for $i=3,\dots,6$ and some $d_1,\dots, d_6 \in [0,1]$.  
Then, for every word $w = \$_0\$_i\$_j$ where $i=1,2$ and $j=  3,\dots,6$, we have:
\[\mach{P}_{\$}(w) = d_i \times d_j.\]  
By the definition of \pfa{s}, it follows that:
\[d_1 + d_2 = d_3 + d_4 + d_5 + d_6 = 1.\]  
Accordingly, $\mach{P}_{\$}$ is a \lmr{4} solution to $\mach{A}_{\$}$ if and only if $d_i = \frac{1}{2}$ for all $i =1,\dots,6$, which is equivalent to $\mach{P}_{\$}$ being simple.  
Hence, property \cref{prop:c2} holds for $\mach{A}_{\$}$, and the probability values of the transitions $(q'_0,\$_0,x)$ and $(q'_0,\$_0,y)$, which are shared by other automata, must be $\frac{1}{2}$.  

Next, consider $\mach{A}_{\text{ext}}$ and $\mach{A}_{\text{nondet}}$. By definition, if we ignore their accepting states, $\mach{A}_{\text{ext}}$ is a sub-automaton of $\mach{A}_{\text{nondet}}$. For every $p \in Q_{\text{nondet}}$, $\mach{A}_{\text{nondet}}$ contains a structure as shown in \cref{fig:nondet_structure}, where the sub-structure from $p$ is analogous to the sub-structure from $x$ or $y$ in \cref{fig:L_dollar}. Since the probability value of the transition $(q'_0,\$_0,y)$ must be $\frac{1}{2}$, the \lmr{4} solution \pfa based on $\mach{A}_{\text{nondet}}$ must be simple if $\mach{A}_{\text{nondet}}$ has a solution. Hence, property \cref{prop:c2} holds for both $\mach{A}_{\text{nondet}}$ and $\mach{A}_{\text{ext}}$.  

Finally, since $\mach{A}_{\text{extra}}$ is deterministic, property \cref{prop:c2} clearly holds. Accordingly, we can conclude that property \cref{prop:c2} holds for $\mach{A'}$. 

\subsection*{Property \texorpdfstring{\cref{prop:c3}}{\getrefnumber{prop:c3}}}

Since we have shown that $\mach{A'}$ satisfies properties \cref{prop:c1} and \cref{prop:c2}, \cref{prop:uniqueness} allows us to assume that the \pfa{s} based on $\mach{A'}$ and its sub-automata are simple for the remainder of this section.

Property \cref{prop:c3} consists of two sub-properties, \cref{prop:c3}.a and \cref{prop:c3}.b, where \cref{prop:c3b} refers only to the language $\Lang_{\text{ext}} = \{\$_0q_0\} \cdot \Sigma^*$ recognized by $\mach{A}_{\text{ext}}$. Let $\mach{P}_{\text{ext}}$ be the simple \pfa based on $\mach{A}_{\text{ext}}$.  
For every word $w \in \Lang_{\text{ext}}$ and accepting run $\pi$ of $\mach{P}_{\text{ext}}$ on $w$, we have $w = \$_0q_0 \cdot w''$ for some $w'' \in \Sigma^*$, and $\pi = \pi' \pi''$, where $\pi' = (q'_0, \$_0, y)(y, q_0, q_0)$ is the sub-run on $\$_0q_0$, and $\pi''$ is the sub-run on $w''$. By the construction of $\mach{P}_{\text{ext}}$, $\pi''$ is an accepting run of the original \pfa $\mach{P}$ on $w''$. Since both transitions $(q'_0, \$_0, y)$ and $(y, q_0, q_0)$ are deterministic in $\mach{P}_{\text{ext}}$, it follows that  
\[
\mach{P}_{\text{ext}}(w) = \mach{P}(w'').
\]
Based on previous results, in the simple \pfa $\mach{P'}$ based on $\mach{A'}$, the transition $(q'_0, \$_0, y)$ has probability $\frac{1}{2}$, and the transition from state $y$ on reading $q_0$ is deterministic. Thus, we obtain  
\[
\mach{P'}(w) = \frac{1}{2} \cdot \mach{P}_{\text{ext}}(w) = \frac{1}{2} \cdot \mach{P}(w'').
\]
Therefore, \cref{prop:c3}.b holds for $\mach{A'}$.

Property \cref{prop:c3}.a concerns the language $\Lang_{\text{ind}}$, which is the union of $\Lang_{\$}$, $\Lang_{\text{extra}}$, and $\Lang_{\text{nondet}}$. In the previous subsection, we have shown that $\mach{A}_{\$}$ is \lmr{4}, and that $\mach{A}_{\text{extra}}$ is deterministic and shares only the transition $(q'_0, \$_0, x)$ with another sub-automaton $\mach{A}_{\$}$ of $\mach{A'}$. Therefore, for every $w \in \Lang_{\$} \cup \Lang_{\text{extra}}$, the simple \pfa $\mach{P'}$ based on $\mach{A'}$ satisfies $\mach{P'}(w) \geq \frac{1}{4}$.

For every $w \in \Lang_{\text{nondet}}$, word $w$ is of the form $\$_0 p \sigma \tau$, where $\tau = (p, \sigma, s)$ is a nondeterministic transition of $\mach{A}$ for some state $s$. In $\mach{A'}$, there is exactly one accepting run $\pi = (q'_0, \$_0, y)(y, p, p)(p, \sigma, s)(s, \tau, q'_f)$ on $w$. Since the transitions $(y, p, p)$ and $(s, \tau, q'_f)$ are deterministic, the probability along $\pi$ is $\frac{1}{2} \cdot 1 \cdot \frac{1}{2} \cdot 1 = \frac{1}{4}$,
which implies $\mach{P'}(w) \geq \frac{1}{4}$. Thus, property \cref{prop:c3}.a holds for $\mach{A'}$. 

We have established that $\mach{A'}$ satisfies properties \cref{prop:c1}--\cref{prop:c3}, so $\mach{A'}$ is \lmr{4} if and only if $\mach{P}_{\geq\frac{1}{2}}$ is \lmr{4} for the given \pfa $\mach{P}$. Consequently, the \lmr{} problem for \nfa is undecidable.
\end{proof}

\section{Deciding Resolvability for Unary Finite Automata}
\label{sec:decidingunary}
\begin{thm}
\label{thm:unarysigma2}
    Positive resolvability is in $\NP^{\coNP}$
    for unary \nfa.
\end{thm}

For a unary NFA a memoryless stochastic resolver induces a unary probabilistic automaton, or equivalently a Markov chain. This is given by the initial state as a row vector $I\in\{0,1\}^d$, a stochastic matrix $P\in[0,1]^{d\times d}$ and the final states as a column vector $F\in\{0,1\}^d$ such that $\mathcal{P}(a^n) = I P^n F$.
We make use of the following lemma, which summarises standard results in Markov chain theory to our needs (see e.g.~\cite[Sec~1.2-1.5]{lawler2010introduction} or~\cite[Sec~1.8]{norris1998markov}).
The lemma shows that there are a sequence of limit distributions that are visited periodically. 
Essentially, we get, for each word length $n$, a set of reachable states, and except for a possible initial sequence, these sets of states repeat periodically. The global period $T$ is the least common multiple of the periods of the states; since the period of a state is at most the size of its SCC, $T$ is at most exponential in $d$. Once in the periodical part, whenever there is a reachable accepting state,  we require there be some reachable accepting state in a bottom strongly connected component (SCC) to ensure the probability is bounded away from zero. 
The key point of the following lemma is to translate this into the support of the limit distributions for these states to be non-zero / zero for reachable states in- / outside of bottom SCCs, so that we only need to consider the support graph of the Markov chain, not the probabilities themselves.
This allows us to guess the support of the resolver in $\NP$ and verify in $\coNP$ that any resolver with this support would suffice.
\begin{lem}
\label{lemma:generaltheory}
    Let $P$ be a Markov chain with $d$ states and a dirac initial distribution $I = \operatorname{Dirac}(q_0)$. 
    There exists $T$, bounded by an exponential in $d$, and limit distributions $\pi^{(0)},\dots,\pi^{(T-1)} \in [0,1]^{d}$, such that for every $k \in \{0,\dots,T-1\}$, $\lim_{n\to\infty} IP^k (P^T)^n = \pi^{(k)}$. Furthermore, for every such $k$ and every state $q$, $\pi^{(k)}(q)$ is non-zero if and only if \begin{enumerate}
        \item $q$ is in a bottom SCC, and 
        \item $q$ is reachable in $P$ from $q_0$ by a path of length $k + T\ell$ for some $\ell \leq d$.
    \end{enumerate}
\end{lem}
\begin{proof}
Let $T_q$ be the period of the state $q$, with $T_q = \gcd\{t \mid \text{length $t$ path from $q$ to $q$ in $P$}\}$, if there is no cycle on state $q$ then let $T_q = 1$. Let $ T= \operatorname{lcm}\{T_q \mid 1\le q\le d \}$ called the global period. The period of each state in an SCC divides the length of any simple cycle in it, thus is at most the number of states within the SCC, hence $T$ is at most the product of the SCC sizes, which is bounded by $2^d$.

We consider $T$ copies of the Markov chain with matrix $B= P^T$, for each $k\in\{0,\dots,T-1\}$, consider the initial distribution $IP^k$. It suffices to show the condition for one such $k$, which we now fix and let $J=IP^k$.

The value $JB^\ell(q')$ represents the probability of moving from the initial state $q$ to $q'$ in a path of length $k+T\ell$ in $P$ from $I$, and is non-zero if and only if there is at least one path of length $k+T\ell$ in $P$ from $q$ to $q'$. 

Let $U\subseteq Q$ be the set of transient states of $B$, and let
$V_1,\dots,V_m\subseteq Q$ be its recurrent SCCs. Reordering states,
we obtain
\[
B=P^T=
\begin{pmatrix}
 A & R_1 & \cdots & R_m \\
 0 & S_1 & \cdots & 0 \\
 \vdots & \vdots & \ddots & \vdots \\
 0 & 0 & \cdots & S_m
\end{pmatrix},
\]
where $A=B[U,U]$ is the submatrix for transient states, $S_i=B[V_i,V_i]$ are the submatrices for recurrent SCCs, and $R_i=B[U,V_i]$ contain the probabilities from transient to recurrent states.

Each $S_i$ is itself an ergodic Markov chain, and thus has a unique stationary distribution $\tau_{i}\in[0,1]^{|S_i|}$ such that $\tau_{i}S_i = \tau_{i}$. 
Let $J_i\in[0,1]^{|S_i|}$ be the restriction of $J$ to the states corresponding to $S_i$, and let $\theta$ be the probability mass associated with this component in $J$, i.e., such that $\sum_{q\in S_i} J_{i}(q) = \theta$. Then we have $\lim_{n\to\infty} J_{i} (S_i)^n = \theta \tau_i$. 

For transient states, the limit distribution is zero, as all probability will eventually leak into the bottom SCCs, that is, we have that $\tau(q') = 0$ for $q' \in U$.

However, we must account for the probability mass starting in $U$ and reaching states in $V_i$. We define $C_i\in [0,1]^{U\times V_i}$ where $C_i(q,q')$ is the probability of first reaching $q'$ starting from $q$, which is $C_i=\sum_{\ell=0}^{\infty}A^\ell R_i$, where $(A^\ell R_i)(q,q')$ is the probability of remaining in
$U$ for $\ell$ steps and then moving to $q'$. Equivalently $C_i$ is the unique matrix satisfying $C_i=AC_i+R_i$.

Thus the probability mass eventually reaching  SCC $V_i$ is formed of the initial transient probability that moves to it, and the probability mass that is already there:
\[ 
\theta_i
    =\sum_{q\in U,\;q'\in V_i}J(q)C_i(q,q')+\sum_{q\in V_i}J(q).
\]

Thus we have $JB^n$ converges to the distribution $\tau$ given by
\[
    \tau(q')=
    \begin{cases}
        0, & q'\in U,\\
        \theta_i\tau_i(q'), & q'\in V_i.
    \end{cases}
\]

Either $\theta_i = 0$, if there is no path in $B$ from a state in $\operatorname{supp}(J)$
to $V_i$, or $\theta_i \neq 0$. In the latter case there is such a path, and we may take
it to be simple, hence of length $\ell < d$. In $P$, this corresponds to a path from $q_0$
to $V_i$ of length $k+T\ell$, since $\operatorname{supp}(J)$ consists of the states
reachable from $q_0$ by a path of length $k$.
\end{proof}

\begin{figure}[t]
\centering
\begin{tikzpicture}[node distance=160pt]
\tikzstyle{state}=[draw,shape=circle,minimum size=20pt]
\tikzset{invisible/.style={minimum width=0mm,inner sep=0mm,outer sep=0mm}}

\node[state,initial,initial text=] (q0) at (0,0) {$q_0$};
\node[state] (q1) at (2,0) {$q_1$};
\node[state] (q2) at (2,-1.8) {$q_2$};
\node[state] (q3) at (4.5,0) {$q_3$};
\node[state] (q4) at (7,0) {$q_4$};

\path
(q0) edge[->] node[above] {$a: 1$} (q1)
(q1) edge[->] node[above] {$a: 1-p$} (q3)
(q1) edge[bend left,->] node[right] {$a: p$} (q2)
(q2) edge[bend left,->] node[left] {$a: 1$} (q1)
(q3) edge[bend left,->] node[above] {$a: 1$} (q4)
(q4) edge[bend left,->] node[below] {$a: 1$} (q3);

\end{tikzpicture}
\caption{A Markov chain, viewed as a unary probabilistic automaton, illustrating \cref{lemma:generaltheory}; here $0<p<1$.}
\label{fig:limit-example}
\end{figure}

\begin{exa}
Consider the Markov chain in \cref{fig:limit-example}, with initial state $q_0$.
The states $q_1,\dots,q_4$ have period $2$ and $q_0$ lies on no cycle, so $T=2$.
The only bottom SCC is $\{q_3,q_4\}$.
The state $q_3$ is reachable from $q_0$ exactly by paths of even length at least $2$,
and $q_4$ exactly by paths of odd length at least $3$.
By \cref{lemma:generaltheory}, $\pi^{(0)}$ is non-zero only on $q_3$ and $\pi^{(1)}$
only on $q_4$, so $\pi^{(0)}=\operatorname{Dirac}(q_3)$ and $\pi^{(1)}=\operatorname{Dirac}(q_4)$.
The states $q_1$ and $q_2$ are reachable at lengths of both parities, but they lie outside
every bottom SCC and so carry no limit mass.
These supports depend only on the support of $P$, not on the value of $p$.
\end{exa}

We now show how to decide whether there is a stochastic memoryless resolver.
\begin{proof}[Proof of \cref{thm:unarysigma2}]
Let $\mach{A} = (\{a\}, Q, q_0,\Delta, F)$ be a unary NFA over the alphabet $\{a\}$, with $|Q| = d$, initial state $q_0$ encoded in a row vector $I \in\{0,1\}^d$ and final states encoded in a column vector $F\in\{0,1\}^d$. We ask if there is a $P\in[0,1]^{d\times d}$ with support at most $\Delta$ and $\lambda$ such that $I P^{n}F \ge \lambda$ whenever $a^n\in \Lang(\mach{A})$.

If $P$ exists, it has some support $S\subseteq\Delta$ and we guess this support of $P$ nondeterministically. 
The choice of support induces a new NFA $\mach{A}_S \subseteq \mach{A}$, in which every edge will be attributed a non-zero probability and $\Lang(\mach{A}_S) = \{a^n \mid IP^nF > 0\}$. We verify, in $\coNP$~\cite{stockmeyer1973word}, that $\Lang(\mach{A}_S) = \Lang(\mach{A})$.
If not, this is not a good support. 

Henceforth, we assume $\Lang(\mach{A}_S) = \Lang(\mach{A})$ and consider the condition to check whether the support $S$ is good. For any $P$ with support $S$ we have $IP^n F > 0$ for $a^n\in \Lang(\mach{A}_S)$. 
It remains to decide whether this probability can be bounded away from zero, that is for some $\lambda > 0$, $IP^n F > \lambda$ whenever $a^n\in \Lang(\mach{A}_S)$.
The condition we describe is independent of the exact choice of $P$ and is based only on the structure of $\mach{A}_S$. 
Using $T$ from \cref{lemma:generaltheory} we decompose $\Lang(\mach{A}_S)$ into  $T$ phases and define for all  $k \in \{0,\dots,T-1\}$ the set $\Lang_k = \{a^{k+Tn}\mid n\in\mathbb{N}\} \cap \Lang(\mach{A})$. Each $\Lang_k$ is either finite, in which case it turns out there is nothing to do, or eventually periodic with period $T$, in which case we use \cref{lemma:generaltheory}'s characterisation of non-zero limit: 
\begin{description}
    \item[$\Lang_k$ is finite] 
There is nothing to decide:  any $P$ with support $S$ induces a lower bound $\min \{ IP^k (P^T)^{n} F\mid n\in\mathbb{N}, a^{k+Tn} \in \Lang_k\} > 0$ on the weight of the words in $\Lang_k$, as the minimum over a finite set of non-zero values. 

\item[$\Lang_k$ is eventually periodic] We check that there exists a path from $q_0$ to some final state $q_f\in F$ of length $k+T\ell$ for some $\ell \leq d$, and that $q_f$ is in some bottom SCC in $\mach{A_S}$.
Then for any $P$ with support $S$, \cref{lemma:generaltheory} entails a limit distribution $\pi^{(k)} \in [0,1]^{d}$ such that $\lim_{n\to\infty} IP^k (P^T)^n = \pi^{(k)}$, in which $\pi^{(k)}(q_f)$ is non-zero. Thus $\lim_{n\to\infty} IP^k (P^T)^n F = \pi^{(k)} F = \epsilon_k > 0$. 
In particular, there exists $n_k$ such that, for all $n' > n_k$, we have $IP^k (P^T)^{n'} F > \epsilon_k / 2$. 
Hence, there is a lower bound on the weights of words in $\Lang_k$: $\min \{ IP^k (P^T)^{n} F \mid n : a^{k+Tn} \in \Lang_k, n \le n_k\}\cup\{\epsilon_k/2\} > 0$, which is a minimum over a finite set of non-zero values.  \\
If there is no such path, then for any $P$ with support $S$, the limit distribution is zero in all final states: the probability tends to zero no matter the exact choice of probabilities.
\end{description}
For a given $k$, it can be decided in polynomial time whether $\mathcal{L}_k$ is finite and whether there is a path of the given type. On a boolean matrix $A_S$ representing the reachability graph of $\mach{A}_S$, the vector $I(A_S)^k$ and boolean matrix $(A_S)^T$ can be computed by iterated squaring in which the reachability questions can be checked.

Thus to summarise, our procedure guesses a support $S\subseteq \Delta$ and the period $T$ in $\NP$, and verifies in $\coNP$ both that $\Lang(\mach{A_S}) = \Lang(\mach{A})$ and there is no $k<T$ in which the reachability query fails.
\end{proof}
In \cref{thm:unaryhardness} we show that the problem is $\coNP$-hard already in the finite-ambiguous case, we present this result alongside the non-unary case in \cref{sec:fnfa} due to commonalities.
 
\section{Unambiguous Finite Automata}
\label{sec:unambiguous}

In this section we provide algorithms for deciding the positive resolvablility problem and $\lmb$-resolvability problem for \unfa. We recall the definition of bad support. 
\paragraph{Bad Support} Given an \nfa $\mach{A}$, a support $S$ is {\em bad} if there is no resolver $\mach{R}$ over $S$ that positively resolves $\A$, that is, $\mach{P}^{\A}_{\mach{R}}$ is not a \mr{} solution to $\mach{A}$ for each $\mach{R}$ over $S$. Equivalently, $S$ is bad if for each resolver $\mach{R}$ over $S$, there is an infinite sequence of words $\{w_i\}_i$ such that $\lim_{i \to\infty} \mach{P}_{\mach{R}}(w_i) = 0$. Essentially, the sequence $\{w_i\}_i$ is a witness for the resolver $\mach{R}$ being  bad.

Unambiguous automata are a special class of \finamb automata and we start this section with a general property related to bad supports for \finamb automata. We will use this to derive the decidability and complexity bounds for the resolvability problems for unambiguous automata as well as general \finamb automata later. %

\begin{lem}\label{lem:bad-edges-increase}
Let $S$ be a support for an \fnfa $\A$ with $\Ll(\mach{A}) = \Ll(\mach{A}_S)$. 
For a word $w$ with $k$ accepting runs, $\pi_1,\dots,\pi_k$,  let $b(w)$ be the least number, over all $i \in [1,k]$, of occurrences of nondeterministic transitions from $S$ in $\pi_i$. 
Then $S$ is a bad support if and only if there is an infinite sequence $\{w_{i}\}_i$ of words such that $\{b(w_{i})\}_i$ is a strictly increasing sequence.%
\end{lem}
\begin{proof}
For the forward direction, let $S$ be a bad support.
By definition, it follows that for any resolver $\mach{R}$ on $S$, there is an infinite sequence $\{w_i\}_i$ s.t. $\lim_{i \to\infty} \mach{P}_{\mach{R}}(w_i) = 0$. Given a word $w$, if there is an accepting run $\pi$ of $w$ on a \pfa, then the acceptance probability of $w$ by the \pfa is bounded below by the acceptance probability along $\pi$. 
In other words, for any resolver $\mach{R}$ over support $S$, with $X_S$, the set of nondeterministic transitions in $S$, if $\epsilon_{min} = \min\limits_{(p,\sigma,q) \in X_S}\mach{R}(p,\sigma,q)$ and if $w \in \Ll(\A_S)$ is a word which has an accepting run containing at most $T$ nondeterministic transitions from $S$, then $\mach{P}_{\mach{R}}(w) \geq \epsilon_{min}^T$. It follows that, $\{w_i\}_i$ has an infinite sub-sequence $\{w_{i_j}\}_j$ s.t. $\{b(w_{i_j})\}_{j}$ is a strictly increasing sequence.

For other direction, let  $\{w_i\}_i$ be an infinite sequence s.t. $\{b(w_{i})\}_{i}$ is a strictly increasing sequence. 
Let the ambiguity of $\A$ be $M$. For any resolver $\mach{R}$ on $S$, let $\epsilon_{max} = \max\limits_{(p,\sigma,q) \in X_S}\mach{R}(p,\sigma,q)$. Then $\mach{P}_{\mach{R}}(w) \leq M\epsilon_{max}^{b(w)}$. Since $\{b(w_{i})\}_{i}$  is strictly increasing, it follows that $\lim_{i \to\infty} \mach{P}_{\mach{R}}(w_i) = 0$ 
\end{proof}
We immediately get as a corollary, that positive resolvability is determined by the support and is agnostic of the exact probabilities in a resolver for \finamb automata.
Note that the forward direction of \cref{lem:bad-edges-increase} holds for any automaton and this does not require the automaton to be \finamb. However, the other direction requires the automaton to be \finamb and it may not hold in general as demonstrated in \cref{fig:infAmbi}.
This automaton is not \finamb. This is because any word $a^n$ has exactly $2^n$ accepting runs, since all states are accepting and after every input letter there are exactly 2 choices for transition. Now on the full support (with all edges in the automaton), for each $w_n$ in the sequence of words $\{w_n = a^n\}_{n \in \mathbb{N}}$,  $b(w_n) = n$, which is increasing. However, $\mach{P}_{\mach{R}}(w_n) = 1$ for every resolver $\mach{R}$, since all runs of any word are accepting.

\begin{figure}[t]
\centering
\begin{tikzpicture}[node distance=160pt]
\tikzstyle{state}=[draw,shape=circle,minimum size=20pt]
\tikzset{invisible/.style={minimum width=0mm,inner sep=0mm,outer sep=0mm}}

\node[state,initial,accepting,initial text=] (p) at (0,0) {$q_0$};
\node[state,accepting] (qf) at (3,0) {$q_f$};

\path
(p) edge[bend left, ->] node[above] {$a$} (qf)
(p) edge[loop above] node {$a$} (p)
(qf) edge[loop above] node {$a$} (qf)
(qf) edge[bend left,->] node[below] {$a$} (p);

\end{tikzpicture}
\caption{An automaton that is not \finamb.}
\label{fig:infAmbi}
\end{figure}
\begin{cor}
If $S$ is not a bad support, then $\A$ is positively resolvable with any resolver $\mach{R}$ over $S$.     
\end{cor}

For a \unfa---where each word has exactly one accepting run---removing a transition without changing the language implies that the transition either originates from an unreachable state or leads to a state with empty language. Such transitions can thus be safely removed or assigned probability zero.  
Assuming the \unfa is trim, its support is unique. It is not positively resolvable if and only if this support is bad.

It is easy to see that a \unfa is not positively resolvable if it admits an accepting run of the form 
$q_0 \xrightarrow{x} q \xrightarrow{y} q \xrightarrow{z} f$,  
where $q_0$ is the initial state, $f$ is an accepting state, and there is a nondeterministic transition taken while reading the infix $y$.  
By pumping $y$, we obtain a sequence of accepted words of the form $xy^iz$, each with a unique accepting run, but with decreasing probabilities assigned by any stochastic resolver.  
This follows from the fact that the loop induced by $y$ contains a nondeterministic transition, which causes the probability mass to diminish.
The \unfa in \cref{fig:eg_nfa_b} illustrates this phenomenon: it is unambiguous but not positively resolvable.  
This is witnessed by the word $bb$, where the first $b$ corresponds to the nondeterministic, pumpable part of the word.  

In the following, we show that the existence of such a witness word is both a necessary and sufficient condition for a \unfa not to be positively resolvable.
\begin{lem}\label{lem:ufa-bad-word}
A \unfa $\A$ is not positively resolvable if and only if there exists a word $w = xyz \in \Ll(\A)$ such that the accepting run on $w$ satisfies the following conditions:
\begin{enumerate}
    \item After reading the prefix $x$, the run reaches a state $q \in \delta_{\A}(q_0, x)$, and after subsequently reading $y$, it returns to the same state, i.e., $q \in \delta_{\A}(q, y)$. Moreover, reading $z$ from $q$ leads to an accepting state: $\delta_{\A}(q, z) \cap F \neq \emptyset$.
    \item The words $y$ and $z$ begin with the same letter $a$, it follows that the transition taken when reading $y$ from $q$ is nondeterministic.
\end{enumerate}
\end{lem}
\begin{proof}
It is straightforward that the existence of such a word implies that $\A$ is not positively resolvable.
To prove the converse, suppose that $\A$ is not positively resolvable. 
Assume further that all states in $\A$ are reachable and productive (i.e., from each state, there exists a path to an accepting state).

By \cref{lem:bad-edges-increase}, there must exist a sequence of accepted words whose unique accepting runs contain an increasing number of nondeterministic transitions. 
Since the number of nondeterministic transitions in $\A$ is finite, the pigeonhole principle implies that some nondeterministic transition $q \xrightarrow{a} q'$ must appear infinitely often in this sequence. 
Hence, this transition must lie within a loop.

Let $x$ be a word such that $\delta_{\A}(q_0, x) \ni q$, and let $y$ be a loop word from $q$ such that it starts with the transition $q \xrightarrow{a} q'$ and returns to $q$.
Since the transition $q \xrightarrow{a} q'$ is nondeterministic and no state is unproductive, we can find a suffix $z$ such that starting from $q$, there exists an alternative run reading $a$ to some state $q'' \neq q'$ that leads to an accepting state, i.e., $\delta_{\A}(q, z) \cap F \neq \emptyset$.
Thus, the word $w = xyz$ witnesses $\A$ not being positive resolvable.
\end{proof}

By guessing this witnessing word of \cref{lem:ufa-bad-word} letter by letter, we have an NL procedure for deciding whether a \unfa is positively resolvable or not. 
\begin{thm}\label{thm:1-ambiguous-non-fixed}
The positive resolvability problem for {\unfa} is in NL.
\end{thm}
\begin{proof}
We have an NL procedure to guess the ``bad'' word of \cref{lem:ufa-bad-word} letter by letter.
Since this bad word has a unique run, it consists of an initial path to a state, a loop on this state that begins with a letter $a$, and then a path to a final state, which also begins with the same letter $a$.
\end{proof}

We can show a matching lower bound in the unambiguous case by a straightforward reduction from the emptiness problem for \unfa, which is known to be NL-complete.

\begin{thm}\label{thm:ufa-NL-hard}
Given a \unfa, it is NL-hard to decide if it is positively resolvable, and NL-hard to decide if it is $\lambda$-resolvable for any fixed $\lambda > 0$.
\end{thm}
\begin{proof}
We reduce from the emptiness problem of unambiguous automata which is known to be NL-complete. 
Given an unambiguous NFA $\A$, we create a new unambiguous NFA $\mathcal{B}$ by adding a state $t$ and making it the new initial state. 
We also add a fresh letter $\$$ and two $\$$-transitions, one takes $t$ back to $t$ and another takes $t$ to the initial state of $\A$.
It is also not hard to see that $\mathcal{B}$ is positively resolvable iff $\A$ has an empty language.
Furthermore, the unambiguous NFA $\mathcal{B}$ is positively resolvable iff it is $\lambda$-resolvable for any fixed $\lambda > 0$.
This proves the theorem.\qedhere
\end{proof}

For a positively resolvable \unfa $\A$, \cref{lem:ufa-bad-word} essentially states that no accepting run may contain nondeterministic transitions in its looping part.
This is equivalent to requiring that all transitions within SCCs of $\A$ are deterministic.
Exploiting this structural property, we can compute the largest possible value of $\lambda$ such that $\A$ is $\lambda$-resolvable.
To do so, we collapse all SCCs into single nodes, yielding a rooted directed acyclic graph (DAG), and reduce the problem to counting the maximum number of disjoint paths in this DAG, which can be computed in polynomial time. 

\begin{thm}\label{thm:ufa-lambda-resolvable}
It can be decided in polynomial time whether a \unfa is $\lambda$-resolvable.
\end{thm}

\begin{proof}

Let $\A$ be a \unfa over the alphabet $\Sigma$.
By \cref{thm:1-ambiguous-non-fixed}, we can decide in NL whether $\A$ is positively resolvable.
If so, we compute a threshold value $\lambda'$ such that $\A$ is $\lambda'$-resolvable, but not $(\lambda' + \epsilon)$-resolvable for any $\epsilon > 0$. This allows us to conclude whether $\A$ is $\lambda$-resolvable by checking if $\lambda' < \lambda$.

To compute $\lambda'$, we first eliminate all unreachable states and states with empty language from $\A$ and then perform an SCC decomposition.
For a state $q$, let $C(q)$ be the SCC containing $q$.
The SCC decomposition yields a directed acyclic graph (DAG), which we use to process the SCCs in reverse topological order. We retain the original states and transitions when evaluating branching choices, rather than merging the outgoing transitions of different states within an SCC.

Observe that a resolver starts with probability 1 at the initial state. Whenever a nondeterministic choice is encountered, it distributes the probability among the outgoing transitions.
Suppose that, by fixing one input letter at each state, we obtain $t$ distinct paths from the initial state to sink states. Since the probabilities of these paths sum to at most 1, at least one path must have probability at most $\frac{1}{t}$. This is an upper bound on the achievable resolvability threshold and this bound can be attained by a suitable choice of transition probabilities, making this the optimal resolvability.
Moreover, by \cref{lem:ufa-bad-word}, all transitions within an SCC are deterministic.
Thus, any two states within an SCC can reach each other with probability 1, and the optimal resolvability threshold is the same from every state of an SCC.

We define a function $g$ on the SCCs, where $g(C)$ measures the maximum number of paths from any state in $C$ to sink states branched for some choice of letter.
We will show that the optimal threshold is $\lambda' = \frac{1}{g(C(q_0))}$, where $C(q_0)$ is the SCC containing the initial state.

For each non-sink SCC $C$, we also record a state-letter pair in a function $f$ where $f(C)=(q,a)$ is the state in $C$ and the choosen letter $a$ from $q$ attaining the maximum defining $g(C)$. This choice witnesses the upper bound on the achievable threshold.

We show below that $f$ and $g$ can be computed in polynomial time in reverse topological order of the SCC DAG.
For each SCC $C$, define
\[
g(C)=
\begin{cases}
1, & \text{if } C \text{ is a sink SCC},\\[2mm]
\displaystyle
\max_{\substack{q\in C,\ a\in\Sigma\\
\varnothing\neq\delta(q,a)\subseteq Q\setminus C}}
\sum_{(q,a,q')\in\Delta}g(C(q')),
& \text{otherwise}.
\end{cases}
\]

For every non-sink SCC $C$, let $f(C)=(q,a)$ be a pair attaining the maximum.

Since all transitions within an SCC are deterministic, the value of $g$ is constant within each SCC. Moreover, all successors appearing in the sums belong to already processed SCCs. Hence, $f$ and $g$ can be computed in polynomial time.

Finally, we construct a stochastic resolver $\mathcal{R}_g$
on the original transitions of $\A$ by setting,
for every $(q,a,q')\in\Delta$,
\[
\mathcal{R}_g(q,a,q') =
\frac{g(C(q'))}
{\sum\limits_{(q,a,q'')\in\Delta}g(C(q''))}.
\]

By the definition of $g$, we have
\[
\mathcal{R}_g(q,a,q')\geq
\frac{g(C(q'))}{g(C(q))}.
\]

Now, consider any word $a_1a_2\cdots a_n$ with an accepting run $q_0\xrightarrow{a_1}q_1\cdots\xrightarrow{a_n}q_n$. Its probability under $\mathcal{R}_g$ is at least
\[
\prod_{i=0}^{n-1}
\frac{g(C(q_{i+1}))}{g(C(q_i))}
=
\frac{g(C(q_n))}{g(C(q_0))}
\geq
\frac{1}{g(C(q_0))}.
\]
Thus, $\mathcal{R}_g$ achieves the threshold $1/g(C(q_0))$.

Conversely, consider any resolver $\mathcal R$.
We show by induction over the SCC DAG that, for every state $q$, there is a word accepted from $q$ with probability at most $1/g(C(q))$ under $\mathcal R$.

If $g(C(q))=1$, the claim follows by choosing any accepted word from $q$.

Otherwise, let $f(C(q))=(r,a)$.
Since $q$ and $r$ belong to the same SCC, there is a word $u$ taking $q$ to $r$ deterministically.

From definition of $g$, we have
\[
\mathcal R(r,a,q')\leq
\frac{g(C(q'))}{g(C(q))}.
\]

Since $q'$ belongs to a successor SCC of $C(q)$, the induction hypothesis gives a word accepted from $q'$ with probability at most $1/g(C(q'))$. 

By unambiguity, the resulting word is accepted with probability at most $\mathcal R(r,a,q')\cdot\frac{1}{g(C(q'))}\leq\frac{1}{g(C(q))}$.
Hence $\frac{1}{g(C(q_0))}$ is the optimal threshold.
\end{proof}

\begin{figure}[t]
\centering
\resizebox{\linewidth}{!}{%
\begin{tikzpicture}[
  >=Latex,
  state/.style={circle,draw,minimum size=8mm,inner sep=1pt},
  accept/.style={state,double,double distance=1pt},
  det/.style={->,black,semithick},
  nd/.style={->,red,semithick,densely dotted},
  selected/.style={->,red,semithick},
  lab/.style={font=\small,fill=white,inner sep=1.2pt}
]

\begin{scope}
  
  \node[state]  (r)  at (-1.8, 3.12) {$r$};
  \node[state]  (q1) at ( 1.8, 3.12) {$q_1$};
  \node[state]  (q2) at ( 3.6, 0)    {$q_2$};
  \node[state]  (q4) at ( 1.8,-3.12) {$q_4$};
  \node[accept] (t)  at (-1.8,-3.12) {$t$};
  \node[state]  (q3) at (-3.6, 0)    {$q_3$};

  \draw[->] (2.95,4.00) -- (q1.north east);

  \draw[det] (r) to[bend left=18]
      node[lab,above] {$x$} (q1);
  \draw[det] (q1) to[bend left=18]
      node[lab,below] {$x$} (r);

  \draw[nd] (r) --
      node[lab,pos=.57,left] {$a$} (q3);
  \draw[nd] (r) --
      node[lab,pos=.76,above] {$a$} (q2);

  \draw[nd] (q1) --
      node[lab,pos=.25,above] {$a$} (q3);
  \draw[nd] (q1) --
      node[lab,pos=.55,right] {$b$} (q2);

  \draw[nd] (q2) --
      node[lab,pos=.62,above] {$c$} (q3);
  \draw[nd] (q2) --
      node[lab,pos=.53,right] {$c$} (q4);

  \draw[nd] (q1) to[bend right=11]
      node[lab,pos=.59,left] {$a$} (q4);
  \draw[nd] (q1) to[bend left=11]
      node[lab,pos=.59,right] {$b$} (q4);

  \draw[det] (q3) --
      node[lab,pos=.54,left] {$d$} (t);
  \draw[det] (q4) --
      node[lab,pos=.54,below] {$e$} (t);

  \node at (0,4.35) {$C=\{r,q_1\}$};
\end{scope}

\begin{scope}[xshift=10.7cm]
  
  \node[state]  (pr)  at (-1.8, 3.12) {$r$};
  \node[state]  (pq1) at ( 1.8, 3.12) {$q_1$};
  \node[state]  (pq2) at ( 3.6, 0)    {$q_2$};
  \node[state]  (pq4) at ( 1.8,-3.12) {$q_4$};
  \node[accept] (pt)  at (-1.8,-3.12) {$t$};
  \node[state]  (pq3) at (-3.6, 0)    {$q_3$};

  \draw[->] (2.95,4.00) -- (pq1.north east);

  \draw[det] (pr) to[bend left=18]
      node[lab,above] {$x:1$} (pq1);
  \draw[det] (pq1) to[bend left=18]
      node[lab,below] {$x:1$} (pr);

  \draw[selected] (pr) --
      node[lab,pos=.57,left] {$a:\frac13$} (pq3);
  \draw[selected] (pr) --
      node[lab,pos=.76,above] {$a:\frac23$} (pq2);

  \draw[nd] (pq1) --
      node[lab,pos=.25,above] {$a:\frac12$} (pq3);
  \draw[selected] (pq1) --
      node[lab,pos=.55,right] {$b:\frac23$} (pq2);

  \draw[selected] (pq2) --
      node[lab,pos=.62,above] {$c:\frac12$} (pq3);
  \draw[selected] (pq2) --
      node[lab,pos=.53,right] {$c:\frac12$} (pq4);

  \draw[nd] (pq1) to[bend right=11]
      node[lab,pos=.59,left] {$a:\frac12$} (pq4);
  \draw[selected] (pq1) to[bend left=11]
      node[lab,pos=.59,right] {$b:\frac13$} (pq4);

  \draw[det] (pq3) --
      node[lab,pos=.54,left] {$d:1$} (pt);
  \draw[det] (pq4) --
      node[lab,pos=.54,below] {$e:1$} (pt);

  \node at (0,4.35) {$C=\{r,q_1\}$};
\end{scope}

\end{tikzpicture}%
}
\caption{
(Left) An unambiguous NFA containing the SCC
$C=\{r,q_1\}$. Red dotted edges represent
nondeterministic transitions, while black edges
are deterministic.
(Right) The same automaton with an optimal
stochastic resolver. Red solid edges indicate
the state--letter choices attaining the maximum
defining $g$.
}
\label{fig:ptime-ufa}
\end{figure}

\begin{exa}
See the left of \cref{fig:ptime-ufa} for an illustration
of an unambiguous NFA $\A$ with a nontrivial SCC
$C=\{r,q_1\}$.
We simulate the algorithm by processing the SCCs
in reverse topological order.
Initially, the sink SCC $\{t\}$ is assigned
$g(\{t\})=1$.
Next, we obtain
$g(\{q_3\})=g(\{q_4\})=1$.
Then we have $
g(\{q_2\})=g(\{q_3\})+g(\{q_4\})=2 $.
and set $f(\{q_2\})=(q_2,c)$.
Finally, we process $C=\{r,q_1\}$.
At $q_1$, choosing $a$ yields
$g(\{q_3\})+g(\{q_4\})=2$,
while choosing $b$ yields
$g(\{q_2\})+g(\{q_4\})=3$.
At $r$, choosing $a$ yields
$g(\{q_2\})+g(\{q_3\})=3$.
Thus, $g(C)=3$, and we may choose
$f(C)=(q_1,b)$ or $f(C)=(r,a)$.
Hence, $\A$ is $\frac{1}{3}$-resolvable and optimal stochastic resolver is shown
on the right of \cref{fig:ptime-ufa}.
\end{exa}

\section{Finitely-Ambiguous Finite Automata}
\label{sec:fnfa}

In this section we establish decidability and hardness for both positive resolvability and $\lambda$-resolvability for \finamb \nfa.

\subsection{Hardness of Resolvability}
\label{subsec:fnfa-hardness}
We show that when an \nfa is \finamb, deciding the existence of a resolver, and the existence of a $\lambda$-resolver are both $\coNP$-hard in the unary case and $\PSPACE$-hard otherwise. 
While the unary and non-unary cases differ slightly, both are based on deciding whether the union of $k$ given DFAs are universal. 

The ideas of both proofs are the same, we will extend the given automaton with a new component that ensures it is universal, but the new component itself will not be positively resolvable (see \cref{fig:hardnessidea}). Therefore, if the existing $k$-DFAs are universal then the resolver that uniformly resolves between these $k$-DFAs (assigning zero weight to the new component) is sufficient. However, if the $k$-DFAs are not themselves universal, then the new component must be used to recognise some words, but this component is not positively resolvable, and so neither is the whole automaton. The difference between the proofs will be how we ensure this works even if only a single short word is not in the union of the $k$-DFAs.

\begin{figure}[t]
    \centering
    \scalebox{0.8}{
\begin{tikzpicture}[
    state/.style={circle, draw, minimum size=0.8cm},
    box/.style={rectangle, draw, minimum width=3.5cm, minimum height=1cm, align=center,rounded corners},
    ->, >=stealth, node distance=1.5cm and 0.5cm
  ]

\node[state, initial, initial text=] (q0) {};

\node[box, below right=of q0,yshift=1.3cm] (box1) {Universal but not\\ positively resolvable};
\node[box, right=of box1] (box2) {DFA};
\node[box, right=of box2,xshift=0.5cm] (box3) {DFA};

\node[right=0.2cm of box2, align=center] (dots) {$\cdots$};

\coordinate[above=0.3cm of dots,xshift=-0.2cm] (space);

\draw (q0) edge (box1);
\draw (q0) edge[bend left=10] (box2);
\draw (q0) edge[bend left=11,color=gray] (space);
\draw (q0) edge[bend left=12] (box3);

\end{tikzpicture}
}
\caption{Conceptual idea of hardness in both unary and non-unary cases (exact details differ).}
    \label{fig:hardnessidea}
\end{figure}

In the unary case we reduce from the problem of whether a unary NFA is universal which is $\coNP$-hard \cite{stockmeyer1973word}, but similarly to \cite[Theorem 7.20]{ChistikovKMP22}, we can assume that the shape is the union of simple cycles. By bringing into Chrobak normal form in polynomial time the automaton has an initial stem and a nondeterministic transition to $k$ cycles at the end of the stem. But universality requires that all words are accepting, each state of the stem must be accepting, which can be pre-checked in polynomial time. Thus the hard case remains to determine whether the union of the remaining $k$-cycle DFAs are universal.
\begin{thm}
\label{thm:unaryhardness}
Given a unary k-ambiguous NFA it is $\coNP$-hard to decide if it is positively resolvable, and $\coNP$-hard to decide if it is $1/(k-1)$-resolvable for $1<k\in\ints$.
\end{thm}

\begin{proof}[Proof of \cref{thm:unaryhardness}]

Let $\A$ be a DFA structured as the union of $k$ cycles, reached from a single initial accepting state $q_0$. Let $\mach{B}$ be an NFA, comprised of $\A$ and three additional states. The first $t_1$ is such that $q_0\to t_1$ and $t_1$ is accepting --- this ensures $\mach{B}$ accepts the word of length 1. The states $t_2,t_3$ are used to ensure $\mach{B}$ accepts all words of length $n\ge 2$, but with a particular shape: $q_0\to t_2$, a self-loop on $t_2$ and an edge from $t_2$ to $t_3$. The state $t_3$ has no outgoing edges and is accepting. Note that only one of the $t_1$ or $t_2,t_3$ branch can accept any given word, so the ambiguity increases only by one. Thus, $\mach{B}$, depicted in \cref{fig:hardnessunary}, is universal and $k+1$ ambiguous. 
\begin{figure}[t]
    \centering
    
\begin{tikzpicture}[
    state/.style={circle, draw, minimum size=0.8cm},
    box/.style={rectangle, draw, minimum width=3.5cm, minimum height=1cm, align=center, rounded corners},
    ->, >=stealth, node distance=1.5cm and 0.5cm
  ]

\node[state, initial, initial text=] (q0) {$q_0$};

\node[state, below=of q0,yshift=0.8cm,accepting] (t1) {$t_1$};
\node[state, right=of t1] (t2) {$t_2$};
\node[state, right=of t2,accepting] (t3) {$t_3$};

\node[box, right=of t3] (box2) {DFA (cycle)};
\node[box, right=of box2,xshift=0.5cm] (box3) {DFA (cycle)};

\draw[->] (q0) -- (t1);
\draw[->] (q0) -- (t2);

\draw[->] (q0) edge[bend left=10] node[above left] {} (box2);
\draw[->] (q0) edge[bend left=10] node[above left] {} (box3);

\draw[->] (t2) edge[loop below] node[below] {} (t2);
\draw[->] (t2) edge[]  (t3);

\node[right=0.2cm of box2, align=center] (dots) {$\cdots$};
\coordinate[above=0.3cm of dots,xshift=-0.2cm] (space);
\draw (q0) edge[bend left=10,color=gray] (space);

\end{tikzpicture}
\caption{Automaton showing coNP-hardness of resolvability in the unary case.}
    \label{fig:hardnessunary}
\end{figure}

\begin{claim}
$\A$ is universal if and only if $\mach{B}$ is positively resolvable if and only if $\mach{B}$ is $1/k$-resolvable.
\end{claim}
If $\A$ is universal, the stochastic resolver can assign $1/k$ probability to each branch in $\A$, and assign zero probability to $t_1,t_2,t_3$ branches of $\mach{B}$. Hence, $\mach{B}$ is positively resolvable, in particular $1/k$-resolvable.

Conversely, suppose that $\A$ is not universal, and assume towards a contradiction that some
resolver $\mach{R}$ positively resolves $\mach{B}$ with threshold $\lambda>0$. Due to the cyclic
structure, there are infinitely many words $w_1,w_2,\dots \notin \machlang{\A}$ with
$2 \le |w_1| < |w_2| < \cdots$. These words are accepted by $\mach{B}$ only via the $t_2,t_3$
branch, so $\mach{R}$ assigns positive probability to the edge from $q_0$ to $t_2$, some
probability $p\in(0,1)$ to the self-loop on $t_2$, and $1-p$ to the edge from $t_2$ to $t_3$.
The word of length $n\ge 2$ then has probability at most $p^{n-2}$ along this branch, so
$\mach{P}_{\mach{R}}(w_i) \to 0$ as $i\to\infty$, contradicting
$\mach{P}_{\mach{R}}(w_i) \ge \lambda$ for all $i$.
\end{proof}

In the non-unary case we show $\PSPACE$-hardness via the problem asking whether the union language of $k$-DFAs is universal, which is $\PSPACE$-complete. 
This problem can be seen as the complement of the (well-known) problem of intersection emptiness for DFAs, that is, whether the intersection of $k$ (not fixed) DFAs is empty, which is known to be $\PSPACE$-complete \cite{Kozen1977}. 
The unary case exploits the cyclic structure to ensure that if some word is not in the union, then there are indeed infinitely many missing words. There is no such natural analogue in the non-unary case requiring a modified construction.

\begin{thm}
\label{thm:SR-fixed-k}
Given a $k$-ambiguous NFA, it is $\PSPACE$-hard to decide if it is positively resolvable, and $\PSPACE$-hard to decide if it is $1/(k-2)$-resolvable for $2<k\in\ints$.
\end{thm}

\begin{proof}[Proof of \cref{thm:SR-fixed-k}]
Given $k$ DFAs, we construct a $(k+2)$-ambiguous NFA $\A$, such that the union language of the $k$ DFAs is universal if and only if $\A$ is positively resolvable.

Without loss of generality, assume the DFAs are over $\Sigma = \{a, b\}$.
We introduce four additional states on top of the $k$ DFAs. 
The new state $t_1$ in $\A$ transitions to itself upon reading $a$ and nondeterministically transitions to the initial states of the $k$ DFAs upon reading $b$.
The new states $t_2$ and $t_3$ ensure that $\A$ accepts all words. 
At state $t_2$, there are only $a$-transitions, which lead either to itself or to $t_3$. 
Similarly, at state $t_3$, there are only $b$-transitions, which lead either to itself or to $t_2$.

All three new states, $t_1$, $t_2$, and $t_3$, are made accepting. 
The new initial state $t_0$ transitions to $t_1$, $t_2$, and $t_3$ upon reading $\varepsilon$.

\begin{figure}[t]
    \centering
    
\begin{tikzpicture}[
    state/.style={circle, draw, minimum size=0.8cm},
    box/.style={rectangle, draw, minimum width=3.5cm, minimum height=1cm, align=center, rounded corners},
    ->, >=stealth, node distance=1.5cm and 0.5cm
  ]

\node[state, initial, initial text=] (q0) {};

\node[state, right=of q0,yshift=0.6cm,accepting] (t1) {$t_1$};
\node[state, below=of t1, yshift=1cm,accepting] (t2) {$t_2$};
\node[state, right=of t2,accepting] (t3) {$t_3$};

\node[box, right=of t3] (box2) {DFA};
\node[box, right=of box2,xshift=0.5cm] (box3) {DFA};

\draw[->] (q0) -- (t1);
\draw[->] (q0) -- (t2);
\draw[->] (q0) edge[bend left=10, in=130] (t3);

\draw[->] (t1) edge[loop above] node[above] {$a$} (t1);
\draw[->] (t1) edge[bend left=10] node[above left] {$b$} (box2);
\draw[->] (t1) edge[bend left=20] node[above left] {$b$} (box3);

\draw[->] (t2) edge[loop below] node[below] {$a$} (t2);
\draw[->] (t2) edge[bend left] node[above] {$a$} (t3);
\draw[->] (t3) edge[loop below] node[below] {$b$} (t3);
\draw[->] (t3) edge[bend left] node[below] {$b$} (t2);

\node[right=0.2cm of box2, align=center] (dots) {$\cdots$};
\coordinate[above=0.3cm of dots,xshift=-0.2cm] (space);
\draw (t1) edge[bend left=15,color=gray] (space);

\end{tikzpicture}
\caption{Automaton showing PSPACE-hardness of resolvability in the general case.}
    \label{fig:hardnessgeneral}
\end{figure}

It is not hard to see that any word, upon reading $\varepsilon$, has at most $k$ accepting runs when transitioning to $t_1$ and exactly two accepting runs when transitioning to either $t_2$ or $t_3$, depending on the initial letter of the word. If the initial letter is $a$, there are two accepting runs from $t_2$ and none from $t_3$. 
Conversely, if the initial letter is $b$, there are two accepting runs from $t_3$ and none from $t_2$. Thus, $\A$, depicted in \cref{fig:hardnessgeneral}, is universal and $k+2$ ambiguous.

\begin{claim}
The union of the $k$ DFAs is universal if and only if $\mach{A}$ is positively resolvable if and only if $\mach{A}$ is $1/k$-resolvable.
\end{claim}

Assume that the union language of the $k$ DFAs is universal. 
The stochastic resolver can assign probability $1$ to the $t_1$ branch, probability $0$ to the $t_2$ and $t_3$ branches, and a probability of $\frac{1}{k}$ for each of the $k$ nondeterministic transitions from $t_1$ to the initial states of the $k$ DFAs.
Thus, $\A$ is stochastically resolvable with a threshold of $\frac{1}{k}$.
Indeed, every word is of the form $a^i$ or $a^ibw$ with $i\ge 0$. The word $a^i$ is accepted
with probability $1$ by staying in $t_1$. Since the union is universal, some DFA accepts $w$,
so $a^ibw$ is accepted with probability at least $\frac{1}{k}$ by moving from $t_1$ on $b$ to
that DFA.

Conversely, assume that a word $w$ is missing from the union language of the $k$ DFAs, and
assume towards a contradiction that some resolver $\mach{R}$ positively resolves $\A$ with
threshold $\lambda>0$. Let $w_i = a^ibw$ for $i\ge 1$. No $w_i$ is accepted via $t_1$, since
the only such run moves on $b$ into a DFA, which rejects $w$, and none is accepted via $t_3$,
which has no $a$-transitions. Hence $\mach{R}$ assigns positive probability to the transition
from $t_0$ to $t_2$ and to both $a$-transitions from $t_2$; let $p\in(0,1)$ be the probability
of the self-loop. Every run through $t_2$ on $w_i$ takes the self-loop $i-1$ times before
moving to $t_3$ on the last $a$, so $\mach{P}_{\mach{R}}(w_i) \le p^{i-1}$, which tends to
zero, contradicting $\mach{P}_{\mach{R}}(w_i)\ge\lambda$.
\end{proof}
\subsection{Deciding Positive Resolvability for Finitely-Ambiguous \nfa}
\label{subsec:fnfa-pr}

In \cref{sec:unambiguous}, we established positive resolvability for unambiguous automata by identifying a witnessing ``bad'' word. 
In this section, we generalise the notion of a ``bad'' word to \finamb \nfa.
For any \unfa $\A$, there exists a unique support $S$ such that $\Ll(\A) = \Ll(\A_S)$ %
assuming $\A$ is trim.
However, when $\A$ is no longer unambiguous, there may be multiple such supports $S$ satisfying $\Ll(\A) = \Ll(\A_S)$.
If $\A$ is not positively resolvable, then all of its supports are \emph{bad}, in the sense that no stochastic resolver over any of these supports can positively resolve $\A$.
Our algorithm for deciding positive resolvability works by checking if supports of $\A$ are bad.
We characterise a bad support $S$, %
using a witness \emph{bad} word that satisfies several conditions. 
This is a generalisation of the witnessing word in the \unfa case. 
Intuitively, a bad word for a support is a word containing subwords that can be pumped arbitrarily to produce new words with acceptance probability arbitrarily close to $0$. 
However, finding a word that can be pumped is not sufficient, as pumping can create new runs that do not have nondeterministic transitions. Consider the automaton in \cref{fig:fnfa-bad-pumping}. For the word, $abb$ there is a unique run $q_0 \xrightarrow{a} q_1 \xrightarrow{b} q_1 \xrightarrow{b} f_1$, with a nondeterministic transition on $b$ at $q_1$. However, for the pumped word $ab^ib$ with $i \geq 2$ there are two runs, one of which is $q_0 \xrightarrow{a} q_1 \xrightarrow{b^i} q_1 \xrightarrow{b} f_1$ that will vanish with increasing $i$, and the other is $q_0 \xrightarrow{a} q_2 \xrightarrow{b} q_3 \xrightarrow{b} q_4 \xrightarrow{b} f_2 \xrightarrow{b^{i-2}} f_2$ with no vanishing nondeterministic transitions. Hence, pumping the subword $b$ does not necessarily yield new words with acceptance probability close to zero.

\begin{figure}[t]
\centering
\begin{tikzpicture}
\tikzset{
xscale=1,>=latex',shorten >=1pt,node distance=2.5cm,
every state/.style={inner sep =.09cm,minimum size=1},
accstate/.style={ circle, fill=gray!30,inner sep =.09cm,minimum size=1},
on grid,auto,initial text = {}	
}

\node[state,initial]  (q0) at (-4.8,1.2) {$q_0$};

\node[state]  (q1) at (-1.8,2.6) {$q_1$};
\node[state,accepting] (f1) at (1.0,2.6) {$f_1$};

\node[state]  (q2) at (-1.8,0.2) {$q_2$};
\node[state]  (q3) at (0.2,0.2) {$q_3$};
\node[state]  (q4) at (2.2,0.2) {$q_4$};
\node[state,accepting] (f2) at (4.2,0.2) {$f_2$};

\draw[->] (q0) -- node[above left,black] {$a$} (q1);
\draw[->] (q1) edge[loop above] node[black] {$b$} ();
\draw[->] (q1) -- node[above,black] {$b$} (f1);

\draw[->] (q0) -- node[below left,black] {$a$} (q2);
\draw[->] (q2) -- node[above,black] {$b$} (q3);
\draw[->] (q3) -- node[above,black] {$b$} (q4);
\draw[->] (q4) -- node[above,black] {$b$} (f2);
\draw[->] (f2) edge[loop above] node[black] {$b$} ();

\end{tikzpicture}
\caption{A 2-ambiguous NFA where pumping naively fails for checking bad support. }
\label{fig:fnfa-bad-pumping}
\end{figure}

Hence, these ``pumpable'' subwords must preserve some reachability behaviour as well as contain probabilistic choices with probability in $(0,1)$. This leads us to the definition of a bad word for a support.

Before defining a bad word, we first introduce some notations we use in the definition. For a word $w \in \Sigma^*$ and $\ell \geq 1$, an \emph{$\ell$-decomposition} of $w$ is a decomposition of the form $w = x_0y_1x_1\dots y_\ell x_\ell$ where $x_0,\dots,x_\ell \in \Sigma^*$ and $y_1,\dots,y_\ell \in \Sigma^+$. For such an $\ell$-decomposition, let $u_i = x_0y_1x_1\dots y_{i-1}x_{i-1}$ for all $i \in [\ell]$. For a support $S \subseteq \Trans$, $\rho \in \accsub{\A_S}{w}$ and a prefix $u$ of $w$, let $\rho(u)$
denote the state of $\rho$ after reading $u$. Let $Q_i = \{\rho(u_i) \mid \rho \in \accsub{\A_S}{w}\}$ and $R_i = \delta_{\A_S}(q_0,u_i) \setminus Q_i$ be the states reached after $u_i$ on some accepting run and on no accepting run, respectively.

\begin{defi}(Bad word)\label{def:bad-word}
Let $\A$ be an \fnfa, $S \subseteq \Trans$ a support, and $w \in \Sigma^*$ a word with $M$
accepting runs in $\A_S$.
$w$ is a \emph{bad word} for $S$ if for some $\ell \leq M$ there is an $\ell$-decomposition of $w$, $w = x_0y_1x_1\dots y_\ell x_\ell$ and states $q_1, \dots, q_\ell$  with $q_i \in Q_i$ for all $i \in [\ell]$, such that:
\begin{enumerate}
  \item every accepting run of $w$ loops on each $y_i$, that is, $\rho(u_i y_i) = \rho(u_i)$ for all
        $\rho \in \accsub{\A_S}{w}$ and $i \in [\ell]$;
  \item reading $y_i$ reaches no new states that lie on no accepting run of $w$, that is,
        $R_i = \delta_{\A_S}(q_0,u_iy_i) \setminus Q_i$ for all $i \in [\ell]$;
  \item every accepting run $\rho$ satisfies $\rho(u_i) = q_i$ for some $i \in [\ell]$, and
        whenever $\rho(u_i) = q_i$, the part of $\rho$ on $y_i$ contains a transition from $S$ that
        is nondeterministic in $\A_S$.
\end{enumerate}
\end{defi}

\begin{exa}\label{exmp:4-ambi-nfa}
The \nfa $\A$ in \cref{fig:fin-ambi-not-sr} is 4-ambiguous.
Consider the full support $S$ containing all transitions of $\A$.
The word $w = \texttt{abbabbac}$ serves as a bad word witness for $S$, with decomposition $w = x_0 y_1 x_1 y_2 x_2 $ where $\ell = 2$, $x_0 = \texttt{ab}$, $y_1 = \texttt{b}$, $x_1 = \texttt{ab}$, $y_2 = \texttt{b}$, and $x_2 = \texttt{ac}$.
There are four distinct accepting runs on $w$; see \cref{fig:bad-word} for an illustration.
Based on \cref{def:bad-word}, for $w$, we have $Q_1 = \{q_1, q_3\}$, $Q_2 = \{q_4, q_6\}$ and $R_1 = R_2 = \{q_f\}$.
Among the four accepting runs, two take nondeterministic transitions from $S$ when reading $y_1$, and the other two do so when reading $y_2$. 
Runs from $Q_i$ to $Q_i$ are highlighted and nondeterministic transitions are coloured in the figure.
This shows that $w$ is a bad word for $S$: the entire sequence of words $\{x_0 y_1^i x_1 y_2^i x_2\}_{i \geq 1}$ admits four accepting runs each, and the probability assigned to each such word diminishes as $i$ increases under any stochastic resolver over $S$.
In fact, this automaton is not positively resolvable. 
Although there exists another support which still preserves the language of $\A$ — the one that excludes the transitions going through $q_2$, we can show that this support is also bad, witnessed by the same bad word. 
With this support, we no longer have the second accepting run in \cref{fig:bad-word}.
\end{exa}

\begin{figure}[t]
\begin{subfigure}[t]{0.37\textwidth}
\centering

\scalebox{0.8}{
\begin{tikzpicture}
\tikzset{
xscale=1,>=latex',shorten >=1pt,node distance=2.5cm,
every state/.style={inner sep =.09cm,minimum size=1},
accstate/.style={ circle, fill=gray!30,inner sep =.09cm,minimum size=1},
on grid,auto,initial text = {}	
}

\node[state,initial left] 	   (s)  {$s$};
\node[state, above right = 2 of s]      (q1) {$q_1$};
\node[state, below right=2 of s]  (q3) {$q_3$};
\node[state, right= 2 of q1] (q4) {$q_4$};
\node[state, right= 2 of q3] (q6) {$q_6$};
\node[state, below= 1.4 of q1] (q2) {$q_2$};
\node[state, right= 2 of q2] (q5) {$q_5$};
\node[state, accepting, below right= 2 of q4] (qf) {$q_f$};

\path[->] (s) edge  [left]node {\color{black} $a$} (q1);
\path[->] (s) edge  [below]node {\color{black} $a$} (q3);
\path[->] (s) edge  [below]node {\color{black} $a$} (q2);

\path[->] (q1) edge  [loop above,color = red,  above]node {\color{black} $b$} (q1);
\path[->] (q1) edge  [above ]node {$a$} (q4);
\path[->] (q1) edge  [above, ]node {\color{black} $b$} (qf);

\path[->] (q3) edge  [above]node {$a$} (q6);
\path[->] (q4) edge  [above]node {$a$} (qf);
\path[->] (q6) edge  [right]node {\color{black} $a,b$} (qf);

\path[->] (q2) edge  [left]node {$b$} (q1);
\path[->] (q6) edge  [left]node {\color{black} $a$} (q5);
\path[->] (q5) edge  [below]node {$c$} (qf);

\path[->] (q4) edge  [loop above, above]node {$b$} (q4);
\path[->] (q3) edge  [loop below, below]node {$b$} (q3);
\path[->] (q6) edge  [loop below ,color = blue, below]node {\color{black} $b$} (q6);
\path[->] (qf) edge  [loop above , above]node {$c$} (qf);

\path[->] (q5) edge  [loop left, left]node {$a$} (q5);

\end{tikzpicture}
}
 \caption{Non-\mr{} \fnfa.}
\label{fig:fin-ambi-not-sr}
\end{subfigure}
\hspace{-0.5cm}
\begin{subfigure}[t]{0.63\textwidth}
\centering

\scalebox{1.0}{
\begin{tikzpicture}
\tikzset{
xscale=1,>=latex',shorten >=1pt,node distance=2.5cm,
every state/.style={inner sep =.09cm,minimum size=1},
accstate/.style={ circle, fill=gray!30,inner sep =.09cm,minimum size=1},
on grid,auto,initial text = {}	
}
\usetikzlibrary{positioning,arrows,automata}
\tikzstyle{BoxStyle} = [draw, circle, fill=black, scale=0.4,minimum width = 1pt, minimum height = 1pt]

\tikzset{
  redbox/.style={
    rectangle,
    rounded corners,
    fill=gray!20,
    fill opacity=0.7,   %
    draw=none,
    minimum width=1.7cm,
    minimum height=2.5cm,
    align=center,
    font=\small
  }
}
\tikzset{
    bluebox/.style={
    rectangle,
    rounded corners,
    fill=gray!20,
    fill opacity=0.7,   %
    draw=none,
    minimum width=1.7cm,
    minimum height=2.5cm,
    align=center,
    font=\small
  }
}

\node[label] (a)  at (0.5,0.5) {$\texttt{a}$};
\node[label] (ab) at (1.5,0.5){$\texttt{b}$};
\node[label, color=black] (abb) at (2.5,0.5){\color{red}{${\texttt{b}}$}};
\node[label] (abba) at (3.5,0.5){$\texttt{a}$};
\node[label] (abbab) at (4.5,0.5){$\texttt{b}$};
\node[label, color=black] (abbabb) at (5.5,0.5){\color{blue}${\texttt{b}}$};
\node[label] (abbabba) at (6.5,0.5){$\texttt{a}$};
\node[label] (abbabbac) at (7.5,0.5){$\texttt{c}$};

\node[label] (l1)  at (0,0) {$s$};
\node[label] (l1a) at (1,0){$q_1$};
\node[label, color=black] (l1ab) at (2,0){${q_1}$};
\node[label, color=black] (l1abb) at (3,0){${q_1}$};
\node[label] (l1abba) at (4,0){$q_4$};
\node[label, color=black] (l1abbab) at (5,0){${q_4}$};
\node[label, color=black] (l1abbabb) at (6,0){${q_4}$};
\node[label] (l1abbabba) at (7,0){$q_f$};
\node[label] (l1abbabbac) at (8,0){$q_f$};

\draw[->] (l1) -- (l1a);
\draw[->] (l1a) -- (l1ab);
\draw[->, color=red, line width=0.25mm] (l1ab) -- (l1abb);
\draw[->] (l1abb) -- (l1abba);
\draw[->] (l1abba) -- (l1abbab);
\draw[->, color=black] (l1abbab) -- (l1abbabb);
\draw[->] (l1abbabb) -- (l1abbabba);
\draw[->] (l1abbabba) -- (l1abbabbac);

\node[label] (l2)  at (0,-.5) {$s$};
\node[label] (l2a) at (1,-.5){$q_2$};
\node[label, color=black] (l2ab) at (2,-.5){${q_1}$};
\node[label, color=black] (l2abb) at (3,-.5){${q_1}$};
\node[label] (l2abba) at (4,-.5){$q_4$};
\node[label, color=black] (l2abbab) at (5,-.5){${q_4}$};
\node[label, color=black] (l2abbabb) at (6,-.5){${q_4}$};
\node[label] (l2abbabba) at (7,-.5){$q_f$};
\node[label] (l2abbabbac) at (8,-.5){$q_f$};

\draw[->] (l2) -- (l2a);
\draw[->] (l2a) -- (l2ab);
\draw[->, color=red, line width=0.25mm] (l2ab) -- (l2abb);
\draw[->] (l2abb) -- (l2abba);
\draw[->] (l2abba) -- (l2abbab);
\draw[->, color=black=] (l2abbab) -- (l2abbabb);
\draw[->] (l2abbabb) -- (l2abbabba);
\draw[->] (l2abbabba) -- (l2abbabbac);

\node[label] (l3)  at (0,-1.0) {$s$};
\node[label] (l3a) at (1,-1.0){$q_3$};
\node[label, color=black] (l3ab) at (2,-1.0){${q_3}$};
\node[label, color=black] (l3abb) at (3,-1.0){${q_3}$};
\node[label] (l3abba) at (4,-1.0){$q_6$};
\node[label, color=black] (l3abbab) at (5,-1.0){${q_6}$};
\node[label, color=black] (l3abbabb) at (6,-1.0){${q_6}$};
\node[label] (l3abbabba) at (7,-1.0){$q_5$};
\node[label] (l3abbabbac) at (8,-1.0){$q_f$};

\draw[->] (l3) -- (l3a);
\draw[->] (l3a) -- (l3ab);
\draw[->, color=black] (l3ab) -- (l3abb);
\draw[->] (l3abb) -- (l3abba);
\draw[->] (l3abba) -- (l3abbab);
\draw[->, color=blue,line width=0.25mm] (l3abbab) -- (l3abbabb);
\draw[->] (l3abbabb) -- (l3abbabba);
\draw[->] (l3abbabba) -- (l3abbabbac);

\node[label] (l4)  at (0,-1.5) {$s$};
\node[label] (l4a) at (1,-1.5){$q_3$};
\node[label, color=black] (l4ab) at (2,-1.5){${q_3}$};
\node[label, color=black] (l4abb) at (3,-1.5){${q_3}$};
\node[label] (l4abba) at (4,-1.5){$q_6$};
\node[label, color=black] (l4abbab) at (5,-1.5){${q_6}$};
\node[label, color=black] (l4abbabb) at (6,-1.5){${q_6}$};
\node[label] (l4abbabba) at (7,-1.5){$q_f$};
\node[label] (l4abbabbac) at (8,-1.5){$q_f$};

\draw[->] (l4) -- (l4a);
\draw[->] (l4a) -- (l4ab);
\draw[->, color=black] (l4ab) -- (l4abb);
\draw[->] (l4abb) -- (l4abba);
\draw[->] (l4abba) -- (l4abbab);
\draw[->, color=blue,line width=0.25mm] (l4abbab) -- (l4abbabb);
\draw[->] (l4abbabb) -- (l4abbabba);
\draw[->] (l4abbabba) -- (l4abbabbac);

\node[label] (R)  at (-0.7,-2.0) {$R$};

\node[label] (l5)  at (0,-2.0) {$\emptyset$};
\node[label] (l5a) at (1,-2.0){$\emptyset$};
\node[label, color=black] (l5ab) at (2,-2.0){${\{q_f\}}$}; 
\node[label, color=black] (l5abb) at (3,-2.0){${\{q_f\}}$};
\node[label] (l5abba) at (4,-2.0){$\emptyset$};
\node[label, color=black] (l5abbab) at (5,-2.0){${\{q_f\}}$};
\node[label, color=black] (l5abbabb) at (6,-2.0){${\{q_f\}}$};
\node[label] (l5abbabba) at (7,-2.0){$\emptyset$};
\node[label] (l5abbabbac) at (8,-2.0){$\emptyset$};

\begin{scope}[on background layer]
        \node[redbox] at (2.5,-1) {};
      \node[bluebox] at (5.5,-1) {};
\end{scope}
\end{tikzpicture}
}
 \caption{Four accepting runs on $w$.} 
\label{fig:bad-word}
\end{subfigure}
\captionsetup{justification=raggedright, singlelinecheck=false}
\caption{A bad word $w = x_0 \textcolor{black}{y_1} x_1 \textcolor{black}{y_2} x_2 = \texttt{ab\textcolor{red}{b}ab\textcolor{blue}{b}ac}$, where $x_0 = \texttt{ab}$, \textcolor{red}{$y_1 = \texttt{b}$}, $x_1 = \texttt{ab}$, \textcolor{blue}{$y_2 = \texttt{b}$}, and $x_2 = \texttt{ac}$. The set $R$ includes all reachable states from which no accepting run exists for the remaining suffix. %
}
\label{fig:xx}
\end{figure}

Detecting bad words involves finding a decomposition as well as suitable $Q_i$'s, $q_i$'s and $R_i$'s.  In order to explicitly track these objects for some word, we consider an automaton $\Gamma_{\A}$, called the \emph{run automaton} of $\A$. Intuitively, a run of a word on $\Gamma_{\A}$, records all the states reached in accepting runs of $w$ as well as the reachable states on non-accepting runs of $w$ in $\A$. 

\paragraph{Run Automata} %
The run automaton of a $k$-ambiguous automaton $\A$, denoted by $\Gamma_{\A}$ is an automaton with 
\begin{enumerate}
	\item  state space $\{ (q_1,\dots,q_k),R) \in Q^k \times 2^Q | \{q_1,\dots,q_k\} \cap  R = \emptyset \}$, i.e. for every state $(\mathbf{q},R)$ in $\Gamma_{\A}$, the set of states in the tuple $\mathbf{q}$ and the set $R$ are disjoint
	\item transitions $((q_1,\dots,q_k), R ) \xra{a} ((q'_1,\dots,q'_k), R' ) $ iff $q_i \xra{a} q_i'$ in $\A$ and $ \forall q \in R, \trans(q,a) \incl R'$.
	\item $((q_0)^k,\emptyset)$ is the initial state  and $\{(\mathbf{q},Q') | \mathbf{q} \in F^k, Q' \cap F = \emptyset\}$ is the set of final states.
\end{enumerate} 
~\\
Note that $\Gamma_{\A}$ preserves the language of $\A$ since no new runs are generated and all existing runs are preserved. 

In order to detect bad words, we will look for a particular type of runs on $\Gamma_{\A}$, called \emph{nice runs}, which are runs that track the behaviour of accepting runs of a word in $\A$ in a way that is consistent with the conditions in \cref{def:bad-word}.
\paragraph{Nice Run}
For a word $w = a_1 \dots a_t \in \Ll(\A)$ with $k$ accepting runs $\rho_1, \dots, \rho_k$ in $\A$, a \emph{nice run} of $w$ is a run on $\Gamma_{\A}$ from initial state $((q_0)^k,\emptyset)$ to final states such that
\begin{enumerate} 
\item on reading $a_1 \dots a_i$, it reaches $((q^i_1,\dots,q^i_k), R_i)$ if the run $\rho_j$ is in $q^i_j$ after reading $a_1 \dots a_i$ in $\A$ and 
\item $q \in R_i$ iff $q$ is reachable from $q_0$ and there is no accepting run of $a_{i+1} \dots a_t$ from $q$.
\end{enumerate}

 Every word in $\Ll(\A)$ with strictly fewer than $k$ accepting runs will have a nice run on $\Gamma_{\A}$, where among $k$ components $\rho_1, \dots, \rho_k$, some two runs will be repeated.
But every word $w$ will have a unique nice run on $\Gamma_{\A}$ up to repetition and reindexing of the $\rho_i$'s. %

Note that a bad word for $S$ is essentially a word in $\Ll(A_S)$ which has a nice run on $\Gamma_{\A_S}$ such that in this nice run, the component of the run on the subword $y_i$ is a cycle in $\Gamma_{\A_S}$ with several accepting runs containing a nondeterministic transition from $S$. In \cref{exmp:4-ambi-nfa} these self-loops happen at states $((q_1,q_1,q_3,q_3),\{q_f\})$ and $((q_4,q_4,q_6,q_6),\{q_f\})$ in the nice run of the bad word in the run automaton. 

The following lemma ensures that for the loop words $y_i$ in \cref{def:bad-word}, 
there is always a unique loop from any state $q \in Q_i$ back to itself.
Assume there are two loops from $q$ to $q$ over the word $y$. 
Let $x$ be a word that there is a run from $q_0$ to $q$ and $z$ a word that there is run from $q$ to a final state. 
Then, the sequence of words $w_i= xy^iz$ admits an unbounded number of accepting runs, contradicting that the automaton is \finamb.
\begin{lem}\label{lem:fnfa-loop-one-path}
In an \fnfa, for any accepting run containing a loop from $q$ to $q$ over the word $y$, it is the unique loop from $q$ to $q$ over $y$.
\end{lem}

A nice run for a bad word contains a cycle corresponding to each subword $y_i$.
The following lemma ensures that we can construct new words by pumping these subwords while preserving the number of accepting runs.
\begin{lem}\label{lem:pump-loop-word-conserve-runs} 
Let $w = xyz$ be a word in $\Ll(\A_S)$ with $M$ accepting runs, whose nice run in $\Gamma_{\A_S}$ has a cycle on the subword $y$. Then for each $j \in \nats$, $xy^jz$ also has $M$ accepting runs. 
\end{lem}
\begin{proof}
 Assume, for contradiction, that for some $j \ge 2$, the word $x y^j z$ admits strictly more accepting runs. Let one such additional accepting run be of the form:
\[
q_0 \xrightarrow{x} q_1 \xrightarrow{y} q_2 \xrightarrow{y} \cdots \xrightarrow{y} q_{j+1} \xrightarrow{z} f,
\]
where $f \in F$ is an accepting state.
Since, in the nice run of $xyz$, the part of this run on the subword $y$ is a cycle, then $xy^jz$ has also a nice run, which just repeats $j$ times the part of the run on $y$. 
Suppose after $x$ and $xy$, the accepting runs are in states $Q_i$ and the reachable states with no continuation to accepting runs of either $yz$ or $z$ be $R_i$.
By assumption, after reading any prefix of the form $xy^k$, the run must be in a state from $Q_i \cup R_i$. Therefore, for all $k \in [j+1]$, we must have $q_k \in Q_i \cup R_i$.  
In particular, $q_{j+1} \in Q_i$, because otherwise $q_{j+1} \in R_i$ contradicting the assumption that all states in $R_i$ do not have accepting runs on $z$.

We show by induction that for all $k \in [j]$, we have $q_k = q_{j+1} \in Q_i$.
\begin{itemize}
    \item \textbf{Base case:} Consider $q_j$. Suppose $q_j \notin Q_i$. Then $q_j \in R_i$, and again contradict the assumption that all states in $R_i$ do not have accepting runs on $yz$. 
    Hence, $q_j \in Q_i$.  
    Furthermore, if $q_j \ne q_{j+1}$, then we obtain an accepting run
    \[
    q_0 \xrightarrow{x} q_j \xrightarrow{y} q_{j+1} \xrightarrow{z} f
    \]
    on $xyz$, which contradicts that any accepting run reaching $q_j$ after reading $x$ must remain in $q_j$ upon reading $y$.

    \item \textbf{Inductive step:} Assume that for some $K \in [j]$, we have $q_k = q_{j+1} \in Q_i$ for all $k \in [K, j]$.  
    Consider $q_{K-1}$. If $q_{K-1} \in R_i$, then again we would obtain an accepting run on $xyz$ via
    \[
    q_0 \xrightarrow{x} q_{K-1} \xrightarrow{y} q_{j+1} \xrightarrow{z} f,
    \]
    contradicting the assumption. So $q_{K-1} \in Q_i$.  
    If $q_{K-1} \ne q_{j+1}$, then the same path gives a new accepting run through a distinct intermediate state, again a contradiction.
     The entire run is of the form:
    \[
    q_0 \xrightarrow{x} q_{j+1} \xrightarrow{y}  \cdots \xrightarrow{y} q_{j+1} \xrightarrow{z} f,
    \]
    which is not a new accepting run. 
    This contradicts the assumption that $x y^j z$ admits additional accepting runs, completing the proof.
    
    Hence, $q_{K-1} = q_{j+1}$.\qedhere
\end{itemize}
\end{proof}

\begin{cor}\label{lem:pump-bad-word-conserve-runs} 
Let $w = x_0y_1x_1 \ldots x_\ell$ be a bad word for support $S$ with $M$ accepting runs. Then for each $(j_1,\dots,j_{\ell}) \in \nats^{\ell}$, $ x_0y_1^{j_1}x_1y_2^{j_2}\ldots y_\ell^{j_\ell}x_\ell$ also has $M$ accepting runs.
\end{cor}

For \fnfa, a support $S$ is bad when either $\Ll(\A_S) \subsetneq \Ll(\A)$ or there is a bad word to witness that the support is bad.
The following is a key lemma of this section, demonstrating that the existence of such a word is necessary and sufficient for a support to be bad.

\begin{lem}[Bad support-bad word]\label{lem:bad-support-bad-word}
Let $\A$ be an \fnfa, and let $S \subseteq \Trans$ be a support.  
Then, $S$ is a bad support if and only if $\Ll(\mach{A}_S) \subsetneq \Ll(\mach{A})$ or there exists a bad word for $S$.
\end{lem}

\begin{proof}
We focus on the case where we do not lose language in $\A_S$, that is, $\Ll(\A_S) = \Ll(\A)$.
First, we show the backward direction. 
Take a bad word $w = x_0y_1x_1 \ldots x_\ell$ for $S$ from \cref{def:bad-word} with $M$ accepting runs.
Let $w_j = x_0y_1^jx_1y_2^j\ldots x_\ell$ produced by pumping the $y_i$ part $j$ times. 
Then, each $w_j$ also has $M$ accepting runs, according to \cref{lem:pump-bad-word-conserve-runs}.
The sequence of words $\{w_j\}_j$ also has an increasing number of nondeterministic transitions across all their accepting runs.
Hence, $S$ is a bad support, by \cref{lem:bad-edges-increase}.

For the forward direction, suppose there is no good resolver over $S$. Let the automaton be $k$-ambiguous. %
From \cref{lem:bad-edges-increase}, we know that there exists a sequence of words which has an increasing number of nondeterministic transitions across all their accepting runs.
For some $M \leq k$, we should find a subsequence of words in this sequence, where each of the words has exactly $M$ accepting runs. 
Let $w$ be a word in $\Ll(\A_S)$ such that every accepting run of $w$ contains at least $T + 1$ nondeterministic transitions, where $T = |Q|^k 2^{|Q|}$. %
Let $w = a_1 \dots a_t$ and the $M$ runs of $w$ be $\pi_1, \dots, \pi_M$. Consider a nice run of $w$ in the run automaton $\Gamma_{\A_S}$. %
Now for an accepting run $\pi_i$ of $w$, we pick a subsequence of states from this run, where run $\pi_i$ takes a nondeterministic transition from $S$. Since, $|\Gamma_{\A_S}| = T$, there are at least two states in this subsequence that are the same. This means, the nice run on $\Gamma_{\A_S}$ has a loop with run $\pi_i$ containing a nondeterministic transition from $S$. Let this subword on this loop be $y_i$, and we have a decomposition $w = x_iy_iz_i$ for each accepting run $\pi_i$. %

We can then look at the words $x_iy_iz_i$ for all $i \in [M]$ and construct a new word $w'=x_0'y_1'x_1' \ldots y_{M}'x_{M}'$ satisfying the conditions of \cref{def:bad-word}. 
We do it the following way: if the $y_i$'s do not overlap, we immediately obtain the decomposition required for a bad word. Otherwise, suppose there is an overlap between $x_{i_1} y_{i_1} z_{i_1}$ and $x_{i_2} y_{i_2} z_{i_2}$. In this case, we consider the word $x_{i_1} y_{i_1} y_{i_1} z_{i_1}$ by pumping $y_{i_1}$ once. Now, $y_{i_2}$ must be either part of the $x_{i_1} y_{i_1}$ segment or the $y_{i_1} z_{i_1}$ segment. Without loss of generality, assume the latter. We then have $y_{i_1} z_{i_1} = x'_{i_2} y_{i_2} z'_{i_2}$, which leads to the word $x_{i_1} y_{i_1} x'_{i_2} y_{i_2} z'{i_2}$, where the looping subwords $y_{i_1}$ and $y_{i_2}$ no longer overlap.
In this way, we can construct a new word $w'$ such that the looping parts $y_i$'s do not overlap. From \cref{lem:pump-loop-word-conserve-runs}, it follows that $w'$ still has $M$ accepting runs. Furthermore, the decomposition of $w'$ satisfies all the conditions in the definition of a bad word.
This proves the lemma.
\end{proof}

To decide the positive resolvability of an \fnfa $\A$, we consider all possible supports $S$ such that $\Ll(\A_S) = \Ll(\A)$, and check whether each support is bad by analysing the run automaton $\Gamma_{\A_S}$.
According to \cref{lem:bad-support-bad-word}, if a support is bad, there must exist a bad word for it, and there must also be a corresponding nice run over this word in $\Gamma_{\A_S}$.
This nice run encodes the structure of all accepting runs of $\A_S$ on the bad word.
In particular, for each accepting run $\pi$ of $\A_S$, the nice run in $\Gamma_{\A_S}$ includes a loop that contains at least one nondeterministic transition from $S$ in its component corresponding to $\pi$.
To show that a support is bad, it suffices to search for such nice runs in $\Gamma_{\A_S}$, which gives us a decidable procedure. A naive algorithm to check for a bad word would be to nondeterministically guess a nice run of a bad word in $\Gamma_{\A_S}$ from $((q_0)^k,\emptyset)$ to $(\mathbf{q},Q')$ with $\mathbf{q} \in F^k$ and $Q' \cap F = \emptyset$. This involves guessing a letter at each step, keeping track of states reached in each run, i.e. at $j$th step computing the states $(q^j_1,\dots,q^j_k)$ and guessing $R_j$ for state $((q^j_1,\dots,q^j_k), R_j)$ in $\Gamma_{\A_S}$. For the bad word decomposition, we would also need to guess the states $Q_i$ and the state $q_i \in Q_i$ for simulating the cycle in $\Gamma_{\A_S}$ with nondeterministic transition in some run. We also track runs that have already seen nondeterministic transition in some cycle on $\A_S$. Finally, we have a bad word, if all runs have seen nondeterministic transitions in some cycle and the sets $R_j$'s are consistent with the final guessed word.

However, given an \fnfa, since its degree of ambiguity $k$ is known to be bounded by $5^{\frac{|Q|}{2}} |Q|^{|Q|}$~\cite{WEBER1991325}, %
the size of the run automaton $\Gamma_{\A_S}$ can be $|Q|^k 2^{|Q|}$ in the worst case, 
 which is doubly exponential in the size of~$\A$. Hence the described procedure is not in $\PSPACE$. By slightly modifying this procedure, instead of storing the state tuple $(q^j_1,\dots,q^j_k)$ reached at every step, we store the set of states $A_j = \{q^j_1,\dots,q^j_k\}$ along with keeping track of those states $G_j$ in $A_j$, to which all accepting runs are yet to see nondeterministic transitions in some cycle in $\Gamma_{\A_S}$. Additionally, for the cycle part, when we guess $Q_i$, since from \cref{lem:fnfa-loop-one-path} we know that every state $q \in Q_i$ has unique run to itself for the subword $y_i$, there are at most $n$ runs to track in this part. So we arbitrarily order the runs and track the progression of these runs in the cycle. Note that, in this part any transition does not change the size of the sets $A_i$ and hence are bijections on $A_i$. We guess these bijections until the compositions of all the bijections is the identity map, i.e. each run has returned to the point of entering the cycle. At the end of the cycle we remove all states from $G_i$, whose cycles have encountered nondeterministic transitions from $S$. 
Since this procedure stores information only about sets $A_j,G_j,R_j$ of size at most $n$, as well as bijections on $[n]$, the state space can be described using $O\big(n^{2n+4}\big)$, giving us a $\PSPACE$ procedure for detecting bad words.

\begin{thm}\label{thm:SR-fnfa-pspapce}
It can be decided in $\PSPACE$ if an \fnfa is positively resolvable. 
Moreover, the shortest bad word is of length at most $O\big(n^{2n+4}\big)$.
\end{thm}
\begin{proof}
We provide a \PSPACE~ algorithm for the complement problem, i.e.\ checking if the given \fnfa is not positive resolvable. 
The first step is to go through all possible supports $S \incl \Trans$ that are language equivalent to $\A$. This can be checked in \PSPACE~\cite{stockmeyer1973word} for general NFA. 

Given a support $S$ with $\Ll(\A_S) = \Ll(\A)$, we check whether it is bad by nondeterministically guessing a bad word $w$. Specifically, we guess a word of the form $w = x_0 y_1 x_1 \dots y_{\ell} x_{\ell}$, satisfying the conditions of \cref{def:bad-word}, letter by letter. During this process, we also guess the start and end points of each $y_i$, and simultaneously track an abstraction of the accepting runs, as well as of all non-accepting runs.

While reading the word $w$, we can be in two modes : when we read the a segment $x_i$, we are in the \emph{branching} mode; wehn we read a segment $y_i$, we are in the \emph{looping}. While reading a word we maintain the following sets and variables:

\begin{enumerate}
    \item  the set R of reachable states from which the remaining suffix has no accepting continuation;
    \item  the set A of states occurring on accepting
    runs of the overall word; and
    \item  the set G of states occurring on accepting
    runs that have not yet encountered a bad
    loop containing a nondeterministic transition.
\end{enumerate}

The sets always satisfy the following conditions: $R \cap A = \emptyset$ and $G \subseteq A$. 

The update rules from sets $R,A,G$ to $R',A',G'$ when guessing a letter $a$ are that all $a$-successors of a state in $R$ must be in $R'$, all $a$-successors of a state in $A$ must be in $A' \uplus R'$ and at least one of them must be in $A'$, and all states in $G'$ must be $a$-successors of states in~$G$.

When reading a looping part $y_i$, we remember the set $R$ and denote as $R_0$,
the abstraction additionally tracks

\begin{enumerate}
    \setcounter{enumi}{3}
\item the set $L$ of states that occurs on some accepting run at the beginning of $y_i$ ($L$ corresponds to the $Q_i$ from \cref{def:bad-word}),
\item a state $q' \in L$ (corresponding to $q_i$ \cref{def:bad-word}) whose loop must contain a nondeterministic transition from $S$,
\item whether this path through $y_i$ has seen a nondeterministic transition using a boolean flag, and
\item a bijection $f\colon L \rightarrow A$.

\end{enumerate}

The bijection $f$ tracks the progression of the accepting runs during the looping part. For each $q \in L$, $f(q)$ is the current state of the run that started at $q$. Initially, $f$ is the identity on $L$. Requiring $f$ to be the identity again at the end of the looping part ensures that every such run returns to its starting state, as required by \cref{def:bad-word}.

Moving from the branching mode to looping mode can be done through an $\varepsilon$-transition, setting $L$ to $A$, $f$ to the identity and $R_0$ to $R$, guessing $q'\in G$ and setting the boolean flag to false.

Moving back to branching mode happens through an $\varepsilon$-transition that can be taken when
the boolean flag is set to true, $f$ is the identity and $R$ is $R_0$.
$f$ being the identity ensures that the accepting runs return to their starting states, while
$R = R_0$ ensures that the set of rejecting reachable states is preserved.
The Boolean flag ensures that the loop through $q'$ contains a nondeterministic transition.
The state $q'$ is then removed from $G$, recording that the corresponding accepting runs have been covered.

The update rules when guessing a letter $a$ in the loop mode additionally require that, for all states $q\in A$ resp.\ $q\in G$, exactly one $a$-successor is in $A'$ resp.\ $G'$; $q'$ is left unchanged, and the boolean flag is set to true if the $a$-transition from $f(q')$ is nondeterministic; otherwise it is left unchanged.
The bijection $f$ is updated to $f'$ such that, if $f \colon q \mapsto p$, then $f' \colon q \mapsto p'$ such that $p' \in A'$ and $p'$ is an $a$-successor of $p$.
We accept the run when we reach a state with $R \cap F = \emptyset$ and $G=\emptyset \neq A \subseteq F$ in the branching mode.

The algorithm essentially reduces to a reachability problem on a graph with an exponential number of nodes. 
However, rather than explicitly constructing this graph, we operate on abstractions involving only a constant number of sets, variables, and functions, thereby requiring only polynomial space.
The algorithm is then in \NPSPACE, hence in \PSPACE, since $\NPSPACE = \PSPACE$~\cite{Savitch1970}.

To show the correctness of the algorithm, we prove that a bad word exists if and only if we can construct an accepting run in the graph.
For a bad word $w = x_0y_1x_1\dots y_\ell x_\ell$ for some $\ell$ with $M$ accepting runs, we observe that the $M$ accepting runs do not have a natural order among them, and that a repetition of each $y_i$ may diminish the probability mass of several accepting runs.
We give each accepting run an index $i$ that refers to the first $y_i$ that diminishes it. 
In particular, the accepting runs with index $i$ are in $q_i'$ after reading $x_0y_1 \ldots y_i$.
We can simply take the terms from \cref{lem:bad-support-bad-word} and construct an accepting run in the graph.

We can likewise construct a bad word from an accepting run on the graph as follows:
starting with $x_0$, the words between two $\epsilon$-transitions define words $x_0,y_1,x_1,y_2,\ldots$; for the $y_i$ we also store their respective $L$ and $q'$ as $Q_i'$ and $q_i'$, respectively.
It is now easy to see that this satisfies all local conditions of \cref{def:bad-word}, so that all we need to show is that all accepting runs are in $q_i'$ after reading $x_0y_1 \ldots y_i$ for some $i$. But this is guaranteed by $G=\emptyset$.

We get as a side result, the shortest bad word cannot be longer than the size of the graph described, and thus, $O\big(n^{2n+4}\big)$.
This is a rough bound: each state can appear in both the domain and range of the bijection, leading to at most $2n$ possibilities. 
There are further $2$ possibilities:
if a state is in the range of the bijection, then it belongs to $A$, and consequently, it can also be in $G$; 
if it is not in the range of the bijection, then it does not belong to $A$ and can either be in $R$ (since $R$ and $A$ are disjoint) or not. 
In addition, it can be the state $q'$ that we need to track in step~5.
It may be in the set $R_0$ or not.
\end{proof}

\subsection{Deciding \texorpdfstring{$\lambda$}{Lambda}-Resolvability for Finitely-Ambiguous \nfa}
\label{subsec:lambda-resolvable}

Here we will provide an algorithm to decide $\lambda$-resolvability for \fnfa. 
We do this by building a finite system of inequalities and check if it is satisfiable. 
Decidability follows from the decidability of first-order theory of reals \cite{Tarski1951}. 
Towards this, we consider \emph{primitive words}.
\paragraph*{Primitive words} A word $w$ is \emph{primitive} for $S$ if a nice run of $w$ in $\Gamma_{\A_S}$ has no cycles.
For example, in \cref{exmp:4-ambi-nfa}, the word $abbabbac$ is not primitive for the given support, while the word obtained by removing the cycles from the nice run - $ababac$ is a primitive word. 

The key idea is that for each primitive word $w$ and corresponding nice run $\pi$ in the run automaton, we can obtain a number $z_w$, which is a lower bound on the probability of accepting any word whose nice run is the same as $\pi$ without the cycles. Moreover, we can find a sequence of such words whose probabilities of acceptance converge to $z_w$.
Hence, the automaton is $\lambda$-resolvable if and only if there is a support 
that satisfies $z_w \ge \lambda$ for every primitive word $w$.

\paragraph*{System of Inequalities}
Since $\A$ is \finamb, we can compute the degree of ambiguity $M$ of $\A$ \cite{CHAN198895}. 
First we check if $\A$ is positively resolvable using the algorithm in \cref{thm:SR-fnfa-pspapce}. 
If yes, for each good support $S$ of $\A$, we build a finite system of inequalities. 

Given a support $S$, we have a variable $x^S_{\tau}$ for each $\tau \in S$, to express the probabilities assigned by some stochastic resolver over $S$.
 
Let $X_S \subseteq S$ be the set of the nondeterministic transitions of $S$.  
The set of constraints $\Theta_S$ ensure that the probability variables satisfy standard probabilistic requirements:
\[
\Theta_S := \Big(\forall \tau \in S. \; 0 < x_\tau^S \le 1\Big) \; \wedge \; \Big(\forall (p, \sigma) \in Q \times \Sigma \text{ of } \A_S. \; \sum_{\tau = (p, \sigma, q) \in S} x_\tau^S = 1 \Big).
\]
These constraints enforce that each transition has positive probability at most $1$, and all outgoing transitions from any state on any letter sum to $1$.

For any word $w = a_1 \dots a_t$ and a run $\pi = \tau_1, \dots, \tau_t$ of $w$ in $\A_S$, we define the probability expression:
\[
p^w_{\pi} := \prod_{\tau_i \in X_S} x^S_{\tau_i}.
\]
This product captures the probability of following the specific run $\pi$, considering only the nondeterministic transitions.

\paragraph*{Primitive Word Constraints}
For each primitive word $w = a_1 \dots a_t$ for support $S$ with $M$ accepting runs $\pi_1, \dots, \pi_M$, consider its nice run on $\Gamma_{\A_S}$. We wish to discard runs of $w$ from our calculation, which can be pumped with non-deterministic transitions for some superword of $w$. Define $\texttt{Bad}(w)$ as the set of accepting runs that would diminish in probability when some superword of $w$ is pumped. Formally, $\pi_j \in \texttt{Bad}(w)$ if there exist words $u, v, u'$ with $v\neq\varepsilon$, and an accepting run $\widehat{\pi}_j$ of $uvu'$ such that:
\begin{enumerate}
    \item $w = uu'$,
    \item $uvu' \in \Ll(\A_S)$,
    \item On the nice run of $uvu'$ in $\Gamma_{\A_S}$, reading $v$ follows a cycle,
    \item Removing the $v$-segment from $\widehat{\pi}_j$ yields $\pi_j$, and this segment contains a nondeterministic transition from $S$.
\end{enumerate}

The set $\texttt{Bad}(w)$ is computed using the same algorithm as in \cref{thm:SR-fnfa-pspapce}. For each primitive word $w$, we define its guaranteed acceptance probability as:
\[
z_w := \sum_{\pi \in \acc{w} \setminus \texttt{Bad}(w)} p^w_{\pi}.
\]
This represents the total probability of accepting runs that do not diminish under pumping.
Each primitive word imposes a constraint:
\[
\varphi_w^S := z_w \geq \lambda.
\]

Let $T = |Q|^{M} 2^{|Q|}$ be the size of the run automaton. Any primitive word must have length at most $T$, since a loop would be formed otherwise. We enumerate all potential primitive words by considering all words $a_1 \dots a_t$ with $t \leq T$, verify which ones are actually primitive using the algorithm from \cref{thm:SR-fnfa-pspapce}, and collect them in $P_S$.
The combined formula for support $S$ is:
\[
\varphi_S := \bigwedge_{w \in P_S} \varphi_w^S
\]

Finally, the complete system of inequalities for deciding $\lambda$-resolvability is:
\[
\varphi := \exists S \subseteq \Delta \; \exists \;\{x^S_{\tau}\}_{\tau \in S}  \,(\Theta_S \land \varphi_S).
\]

This essentially asks, if there exist a support $S$ and probability assignments $\{x^S_{\tau}\}$ such that the probabilistic constraints are satisfied and every primitive word is accepted with probability at least $\lambda$?

\begin{thm}\label{thm:constant-fnfa-lambda-resolvable}
    It is decidable whether an \fnfa is $\lambda$-resolvable.
\end{thm}

\begin{proof}

In the following, we prove that the system of inequalities $\varphi$, has a solution if and only if there exists a resolver $\mach{R}$ such that $\mach{P}_{\mach{R}}(w) \geq \lambda$ for all $w \in \Ll(\A)$.
Suppose there is a solution $\mathbf{x}$ to the system of inequalities. 
Let $\mach{R}[\mathbf{x}]$ be the resolver obtained from this solution with support $S$. We abuse the notation and use $\mach{P}_{\mach{R}[\mathbf{x}]}(\pi)$ to denote the probability of a run $\pi$ of $w$. For each primitive word $w$ for $S$, we have $\sum_{\pi \in \acc{w} \setminus \texttt{Bad}(w)} \mach{P}_{\mach{R}[\mathbf{x}]}(\pi) \geq \lambda$. We will show that for any word $w' \in \Ll(\A)$, either $w'$ is primitive
or there is a primitive word $w$ such that $\mach{P}_{\mach{R}[\mathbf{x}]}(w') \geq \sum_{\pi \in \acc{w} \setminus \texttt{Bad}(w)} \mach{P}_{\mach{R}[\mathbf{x}]}(\pi)$. 

If $w'$ is already primitive, this is already true since runs from $\texttt{Bad}(w')$ are taken out. 
Now if $w'$ is not primitive, then we can find a decomposition $w' = u_1v_1\dots u_tv_t$ such that $w = u_1\dots u_t$ is primitive and a nice run of $w'$ on $\Gamma_{\A_S}$ loops on all subwords $v_i$. We know that $w$ and $w'$ have the same number of runs. 
Since for runs of $w'$ corresponding to run $\pi$ of $w$ that sees nondeterministic transitions from $S$ in any of the subwords from $v_i$'s, $\pi$ is already included in $\texttt{Bad}(w)$. Moreover, the probabilities on nondeterministic transitions for runs of $w$ and $w'$ match on the $u_i$'s.
Thus, $\mach{P}_{\mach{R}[\mathbf{x}]}(w') \geq \sum_{\pi \in \acc{w} \setminus \texttt{Bad}(w)} \mach{P}_{\mach{R}[\mathbf{x}]}(\pi) \geq \lambda$.

For the other direction, let $\mach{R}$ be a stochastic resolver that resolves $\A$ with threshold $\lambda$ over support $S$.
The probabilities assigned by $\mach{R}$ trivially satisfy the constraints in $\Theta_S$.
Now, consider a primitive word $w$.
For each $\pi \in \texttt{Bad}(w)$, there exists a superword $w' = u v_{\pi} u'$, and a nice run of $uv_{\pi}u'$ in $\Gamma_{\A_S}$ looping on $v_{\pi}$. Moreover, the run of $uv_{\pi}u'$ corresponding to $\pi$ contains at least one nondeterministic transition from $S$ in the segment$v_{\pi}$.
We can obtain such superwords $u v_{\pi} u'$ for each $\pi \in \texttt{Bad}(w)$.  
For a run $\pi'$ on $u v_{\pi} u'$, %
we have $\mach{P}_{\mach{R}}(\pi') = \mach{P}_{\mach{R}}(\pi) \cdot y_{\pi}$, where %
$y_{\pi}$ denotes the probability assigned by $\mach{R}$ to the looping segments over all these superwords.
Since the loop run on $v_{\pi}$ has nondeterministic transitions, we have $y_{\pi} < 1$.
Then,
\begin{align*} 
 \mach{P}_{\mach{R}}(w') = \sum_{\pi \in \acc{w} \setminus \texttt{Bad}(w)} \mach{P}_{\mach{R}}(\pi)+ \sum_{\pi \in \texttt{Bad}(w)} \mach{P}_{\mach{R}}(\pi) \cdot y_{\pi} \geq \lambda.
\end{align*}
For each $i \in \mathbb{N}$, let $w_i'$ be the word obtained by modifying $w'$ so that each subword $v_{\pi}$, for every $\pi \in \texttt{Bad}(w)$, is repeated $i$ times. 
We have:
\begin{align*}
\mach{P}_{\mach{R}}(w'_i) = \sum_{\pi \in \acc{w} \setminus \texttt{Bad}(w)} \mach{P}_{\mach{R}}(\pi)+ \sum_{\pi \in \texttt{Bad}(w)} \mach{P}_{\mach{R}}(\pi) \cdot (y_{\pi})^{i} \ge \lambda.
\end{align*}
Hence,
\[
\lim_{i \to \infty}\mach{P}_{\mach{R}}(w'_i) = \sum_{\pi \in \acc{w} \setminus \texttt{Bad}(w)} \mach{P}_{\mach{R}}(\pi) \geq \lambda.
\]
Hence, $\mach{R}$ satisfies all constraints in $\varphi_S$, and thus satisfies the entire system of inequalities.
\end{proof}
\begin{comment}

\begin{proof}[Proof Sketch]
The key idea is that for each primitive word $w$ and corresponding nice run $\pi$ in the run automaton, we can obtain a number $z_w$, which is a lower bound on the probability of any word whose nice run is the same as $\pi$ without the self-loops. Moreover, we can find a sequence of such words whose limit probability is $z_w$.
Hence, the automaton is $\lambda$-resolvable if and only if there is a support 
that satisfies $z_w \ge \lambda$ for every primitive word $w$.

To encode this idea in a formula in the theory of the reals, 
we have a variable $x^S_{\tau}$ for each $\tau \in S$, to express the probabilities assigned by some stochastic resolver over a given support $S$. 
Let $\Theta_S$ denote the set of constraints ensuring that the variables $\{x^S_{\tau}\}_{\tau \in S}$ satisfy the probabilistic constraints and take values strictly within $(0,1]$.

The set $P_S$ of all primitive words associated with the support $S$ can be obtained from the run automaton $\Gamma_{\A}$, along with the corresponding symbolic probability $z_w$. 
We define the formula for the support $S$ as
$
\varphi_S := \bigwedge_{w \in P_S} \varphi_w^S,\text{ where } \varphi_w^S := z_w \ge \lambda \text{ for each } w \in P_S.
$
This gives our final system of inequalities $\varphi := \exists S \subseteq \Delta \; \exists \;\{x^S_{\tau}\}_{\tau \in S}  (\Theta_S \land \varphi_S )$.
\end{proof}
\end{comment}

We conclude by noting that our analysis does not extend to \nfa that are not \finamb, as \cref{lem:bad-edges-increase} fails in this setting; see \cref{fig:infAmbi}.

\begin{comment}
This can be demonstrated by the automaton in \cref{fig:infAmbi}. This automaton is not \finamb. This is because any word $a^n$ has exactly $2^n$ accepting runs, since all states are accepting and after every input letter there are exactly 2 choices for transition. Now on the full support (with all edges in the automaton), for each $w_n$ in the sequence of words $\{w_n = a^n\}_{n \in \mathbb{N}}$,  $b(w_n) = n$, which is increasing. However, $\lim_{n \to \infty} \Pr(w_n) = 1$ since all runs of any word are accepting.

\begin{figure}[t]
\centering
\begin{tikzpicture}[node distance=160pt]
\tikzstyle{state}=[draw,shape=circle,minimum size=20pt]
\tikzset{invisible/.style={minimum width=0mm,inner sep=0mm,outer sep=0mm}}

\node[state,initial,accepting,initial text=] (p) at (0,0) {$q_0$};
\node[state,accepting] (qf) at (3,0) {$q_f$};

\path
(p) edge[bend left, ->] node[above] {$a$} (qf)
(p) edge[loop above] node {$a$} (p)
(qf) edge[loop above] node {$a$} (qf)
(qf) edge[bend left,->] node[below] {$a$} (p);

\end{tikzpicture}
\caption{An automaton that is not \finamb.}
\label{fig:infAmbi}
\end{figure}
\end{comment}
\section{Automata over Infinite Words}
\label{app:discussioninfinite}
In this section we consider how our results extend to $\omega$-regular languages, several of our results transfer elegantly.  The undecidability of $\lambda$-resolvability holds even for weak $\omega$-automata. For finitely-ambiguous case, we can adjust our definition of bad words to \emph{bad prefixes}: where the accepting and rejecting runs in a bad word end in final and non-final states for bad words, for bad prefixes there needs to be some infinite word that is accepted from all states the prefixes of accepting runs end in but from none of the states the prefixes of non-accepting runs end in. Checking if such a word exists can be done in $\PSPACE$ for nondeterministic parity automata (\npa). It follows that the positive resolvability problem for \finamb \npa is \PSPACE-complete. The $\lambda$-resolvability problem is also decidable for \finamb \npa.

A {\em nondeterministic parity automaton} (\npa) $\A$ is
a finite alphabet $\Sigma$,
a finite set of states $Q$,
an initial state $q_0$,
a set of transitions $\Delta\subseteq Q\times \Sigma \times Q$,
a priority function $\prio : Q \mapsto [d]$ for some $d \in \nats$.
An \npa accepts a set of infinite words in $\infwords$. 
A run $\pi$ of $\mach{A}$ on an infinite word $w = \sigma_1\sigma_2 \cdots \in \infwords$ is an infinite sequence of transitions $\tau_1\tau_2\dots$, where each $\tau_i=(p_i,\sigma_i,q_i)$ for $i\in \positives$.
A run $\pi$ naturally generates an infinite trace $\rho(\pi)$ over $[d]$, where the $i$-th element in this sequence is $\rho(p_i)$. Let $\mathit{inf}_{\rho}(\pi)$ be the numbers that occur infinitely often in $\rho(\pi)$.
The run $\pi$ is accepting if $p_1=q_0$ and the maximum element in $\mathit{inf}_{\rho}(\pi)$ is even. 
A {\em nondeterministic \buchi automaton} (\nba) is an \npa where $\prio$ takes values in $\{1,2\}$. A \buchi automaton can be described by providing the set of states with priority $2$ as the accepting set of states. 
A \buchi automaton is weak iff every SCC contains either only accepting states or only rejecting states.
Let $\accsub{\mach{A}}{w}$ denote the set of all accepting runs of $\mach{A}$ on $w$.
The notions of \kamb{k} and \finamb naturally extends to \npa based on the cardinality of $\accsub{\mach{A}}{w}$. %

A word $w \in \infwords$ induces a probability measure $\mathit{Prob}^{\mach{P}}_w$ over the measurable subsets of $\Trans^{\omega}$, defined via the $\sigma$-algebra generated by basic cylinder sets in the standard way. The acceptance probability of $w$ in $\mach{P}$ is then given by:
\[
\mach{P}(w) = \mathit{Prob}^{\mach{P}}_w(\acc{w}).
\]

Definitions of positive resolvability and $\lambda$-resolvability extends naturally to \npa.

\subsection{Undecidability of $\lambda$-Resolvability}
We have shown that the $\lambda$-resolvability problem for automata over finite words (i.e., \nfa) is undecidable. In this subsection, we extend this result to automata over infinite words by padding the end of the finite word with $\sharp^\omega$.

Given an \nfa $\mach{A} = (\Sigma, Q, q_0, F, \Delta)$, define the  \buchi automaton $\mach{A}_\omega = (\Sigma', Q', q_0, F', \Delta')$ as follows:
\begin{itemize}
    \item $\Sigma' = \Sigma \cup \{\sharp\}$, where $\sharp \notin \Sigma$.
    \item $Q' = Q \cup \{q'_f\}$, where $q'_f \notin Q$.
    \item $F' = \{q'_f\}$.
    \item $\Delta' = \Delta \cup \{(p, \sharp, q'_f) \mid p \in F \cup F'\}$.
\end{itemize}

By construction, $\mach{A}_\omega$ has a unique accepting state $q'_f$, which is also a sink state, making $\mach{A}_\omega$ a weak automaton. Suppose $\mach{P}$ is a \pfa based on $\mach{A}$. We construct the corresponding probabilistic automaton $\mach{P}_\omega$ based on $\mach{A}_\omega$ by adding transitions of the form $(p, \sharp, q'_f, 1)$ to $\mach{P}_\omega$ for all $p \in F \cup F'$.
This construction ensures that $\mach{P}(w) = \mach{P}_{\omega}(w \cdot \sharp^\omega)$ for every word $w \in \Sigma^*$. Consequently, for every $\lambda \in (0,1)$, $\mach{A}$ is \lmr{} if and only if $\mach{A}_\omega$ is \lmr{}.
Together with the undecidability of the $\lambda$-resolvability problem for \nfa (\cref{thm:undecidable}), we obtain the following:

\begin{thm}
The $\lambda$-resolvability problem for weak nondeterministic \buchi automata is undecidable.
\end{thm}

\subsection{Decidability for Finitely-ambiguous \npa}
On the other hand, our algorithm for \fnfa also extends to \finamb \npa using similar analysis. %
Since the number of accepting runs of every word is bounded by some $k$, \cref{lem:bad-edges-increase} holds for \finamb \npa as well.

\begin{lem}
Let $S$ be a support for a \finamb{} \npa $\A$ with $\Ll(\mach{A}) = \Ll(\mach{A}_S)$. 
Then $S$ is a bad support if and only if %
    for every $b \in \nats$, there is at least one word $w \in \Ll(\A_S)$ such that $w$ has at least $b$ nondeterministic transitions from $S$ in each of its accepting runs.
\end{lem}

 Hence it is enough to track the behaviour of the \npa up to a finite prefix. Towards this goal, the notion of bad words can be adapted to \npa using \emph{bad prefixes}. 
Let $\A^q$ denote the $\omega$-automaton $\A$ with $q$ as its initial state.

\begin{defi}(Bad prefix)\label{def:bad-prefix}
Let $\A$ be a \finamb \npa and $S \subseteq \Trans$ a support.
For an $\omega$-word $w$, let $\accsub{\A_S}{w}$ be the set of runs of $\A_S$ on $w$ that satisfy
the parity condition, and for $\rho \in \accsub{\A_S}{w}$ and a finite prefix $u$ of $w$, let
$\rho(u)$ denote the state of $\rho$ after reading $u$.
A word $u \in \Sigma^*$ is a \emph{bad prefix} for $S$ if there are an $\omega$-word
$v \in \Sigma^\omega$ with $uv \in \Ll(\A_S)$, a decomposition $u = x_0y_1x_1\dots x_{\ell-1}y_\ell$
with $\ell \geq 1$, $x_0,\dots,x_{\ell-1} \in \Sigma^*$ and $y_1,\dots,y_\ell \in \Sigma^+$, and
states $q_1, \dots, q_\ell$ with $q_i \in Q_i$ for all $i \in [\ell]$, such that the following hold,
where $u_i = x_0y_1\dots y_{i-1}x_{i-1}$,
$Q_i = \{\rho(u_i) \mid \rho \in \accsub{\A_S}{uv}\}$ and
$R_i = \delta_{\A_S}(q_0,u_i) \setminus Q_i$:
\begin{enumerate}
  \item every accepting run of $uv$ loops on each $y_i$, that is, $\rho(u_iy_i) = \rho(u_i)$ for all
        $\rho \in \accsub{\A_S}{uv}$ and $i \in [\ell]$;
  \item reading $y_i$ reaches no new states that lie on no accepting run of $uv$, that is,
        $R_i = \delta_{\A_S}(q_0,u_iy_i) \setminus Q_i$ for all $i \in [\ell]$;
  \item every accepting run $\rho$ of $uv$ satisfies $\rho(u_i) = q_i$ for some $i \in [\ell]$, and
        whenever $\rho(u_i) = q_i$, the part of $\rho$ on $y_i$ contains a transition from $S$ that
        is nondeterministic in $\A_S$.
\end{enumerate}
\end{defi}

\begin{lem}[Bad support-bad prefix]
\label{lem:badSuppNPA}
Let $\A$ be a \finamb \npa and $S$ be a support of $\A$. Then $S$ is a bad support for $\A$ iff $\Ll(\A_S) \subsetneq \Ll(A)$ or there exists a bad prefix for $S$.      
\end{lem}

Hence, our \PSPACE\ algorithm for deciding the positive resolvability of \fnfa can be extended to \finamb \npa by adding an \emph{accepting mode}. 
This mode verifies whether there exists an $\omega$-word $v$ in the set $\bigcap_{q \in A} \Ll(\A^q_S) \setminus \bigcup_{q \in R} \Ll(\A^q_S)$. 
In the branching mode, a transition to the accepting mode is always possible when $G = \emptyset$. Once the system enters the accepting mode, it remains there.

To check whether or not $\bigcap_{q \in A} \Ll(\A^q_S) \setminus \bigcup_{q \in R} \Ll(\A^q_S)$ is empty, we check if $\bigcap_{q \in A} \Ll(\A^q_S) \subseteq \bigcup_{q \in R} \Ll(\A^q_S)$ holds.
To do this, we can simply translate the individual $\A^q_S$ for all $q \in A$ to \nba{s} $\mathcal N_q$ and construct an \nba $\mathcal N$ recognising $\bigcup_{q \in R} \Ll(\A^q_S)$; the blow-up for this is small.

We can now check whether or not $\bigcap_{q \in A} \Ll(\mathcal N_q) \subseteq  \Ll(\mathcal N)$ holds.

For this, we complement $\mathcal N$ using \cite{DBLP:conf/stacs/Schewe09} to the complement \nba $\mathcal C$, but note that this complement -- including the states and transitions -- can be represented symbolically by $\mathcal N$ itself -- it only uses subsets of the states in $\mathcal N$ (the reachable states and a subset thereof used for a breakpoint construction) and level rankings, which are functions from the reachable states of $\mathcal N$ to natural numbers up to $2n$, where $n$ is the number of states of $\mathcal N$ \cite{DBLP:journals/tocl/KupfermanV01,DBLP:conf/stacs/Schewe09}.

We can now check whether or not $\bigcap_{q \in A} \Ll(\mathcal N_q) \cap  \Ll(\mathcal C) = \emptyset$ holds.

To do this, we can search in the exponentially larger, but polynomially represented, product space of $\bigotimes_{q \in A} {\mathcal N_q} \times  \mathcal C$ (the product automaton of ${\mathcal N_q}$ for all $q \in A$ and $\mathcal C$) to check if we can reach a product state, which can loop to itself while seeing an accepting state of each of the automata.
This can be checked in \NL\ in the explicit size of the product automaton and thus in \PSPACE\ in its representation.

\begin{thm}
The positive resolvability problem for \finamb \npa is \PSPACE-complete.
\end{thm}

Let the degree of ambiguity of a \finamb{} \npa $\A$ be $k$, which can be decided \cite[Proposition~21]{RabinovichT21}.
Note that in \cite{RabinovichT21}, the notion of finite ambiguity differs from ours: an automaton is called finitely ambiguous if $|\acc{w}|$ is finite for every word $w$ in the language. 
What we refer to as \finamb{} is instead termed \emph{boundedly ambiguous} in their terminology.

Let $S$ be a support of $\A$, we can construct a run automaton for $\A_S$, similar to the construction for an \nfa{} on a support. 
As in the \nfa{} case, by analysing the run automaton, we can build a finite system of inequalities and decide the $\lambda$-resolvability problem for \finamb{} \npa{}.

\begin{thm}
The $\lambda$-resolvability problem for \finamb \npa is decidable. 
\end{thm}

\section{Conclusion}
We have introduced the notion of $\lambda$-resolvability as a new measure for the degree of nondeterminism. We have shown that checking this property is undecidable for \nfa, and provided further decidability and complexity results for important special cases, in particular finitely-ambiguous automata. 

Our work suggests several directions for further investigation, particularly regarding whether there exist broader classes of \nfa for which positive or $\lambda$-resolvability are decidable. Additional work includes closing the remaining complexity gaps and determining whether checking positive resolvability for general \nfa is decidable. It is also natural to ask how these results extend to more expressive language classes, such as (visible) counter or pushdown languages.

\section*{Acknowledgment}
  This work has been supported by the European Union’s Horizon Europe
(HORIZON) programme under the Marie Skłodowska-Curie grant agreement No 101208673.
We acknowledge the support of the EPSRC through the projects EP/X042596/1 and EP/X03688X/1.
We would also like to thank the anonymous reviewers for their helpful comments and suggestions.

\bibliographystyle{alphaurl}%
\newcommand{\etalchar}[1]{$^{#1}$}

\end{document}